\documentclass[ 
conference
]{IEEEtran}

\usepackage{cite}

\ifCLASSINFOpdf
\else
\fi
\usepackage[cmex10]{amsmath}
\usepackage{enumerate}
\usepackage{bbm}
\usepackage{amsmath}
\usepackage{amsfonts}
\usepackage{graphicx}
\usepackage{epstopdf}
\usepackage{caption2}
\usepackage{subfigure}

\begin{document}
	%
	\title{The impact of bundling licensed and unlicensed wireless service}
	
	\author{\IEEEauthorblockN{Xu Wang and Randall A. Berry}
	\IEEEauthorblockA{EECS Department, Northwestern University, Evanston, IL 60208\\
			Email: xuwang2019@u.northwestern.edu, rberry@eecs.northwestern.edu}\thanks{This research was supported in part by NSF grants AST-1343381, AST-1547328 and CNS-1701921}
	}


	\maketitle

	\begin{abstract}
		Unlicensed spectrum has been viewed as a way to increase competition in wireless access and promote innovation in new technologies and business models. However, several recent papers have shown that  the openness of such spectrum can also lead to it becoming over congested when used by competing wireless service providers (SPs). This in turn can result in the SPs making no profit and may deter them from entering the market.   However, this prior work assumes that unlicensed access is a separate service from any service offered using licensed spectrum. Here, we instead consider the more common case were service providers bundle both licensed and unlicensed spectrum as a single service and offer this with a single price. We analyze a model for such a market and show that in this case SPs are able to gain higher profit than the case without bundling. It is also possible to get higher social welfare with bundling. Moreover, we explore the case where SPs are allowed to manage the customers' average percentage of time they receive service on unlicensed spectrum and characterize the social welfare gap between the profit maximizing and social welfare maximizing setting.
		
		
	\end{abstract}
	

	%
	\IEEEpeerreviewmaketitle

	\newtheorem{definition}{Definition}
	\newtheorem{proposition}{Proposition}
	\newtheorem{corollary}{Corollary}
	\newtheorem{theorem}{Theorem}
	\newtheorem{lemma}{Lemma}
	\section{Introduction}
	The dramatic increase in the population of smart phones and tablets has resulted in a huge increase in the traffic load on existing wireless networks. Adding new unlicensed spectrum to the market has been viewed as a straightforward approach to alleviate the heavy congestion on current licensed bands, as well as provide a way of increasing competition in the market for wireless services. Recent policies in this direction in the U.S. include the unlicensed use of the television white spaces\cite{fcc} and the Generalized Authorized Access (GAA) tier in the 3.5 GHz band\cite{pcast}.

	A key difference of unlicensed spectrum from licensed spectrum is that there is no expensive license cost for the unlicensed band so that more Service Providers (SPs)  are able to enter the market as long as they follow certain technical rules. The increasing competition may benefit the customers and overall social welfare. However, the open access to unlicensed spectrum may lead the band to be over crowded if too many SPs enter the market. Such issues have been studied in \cite{maille2009price,nguyen2011impact,nguyen2015free} by applying a model for price competition with congestible resources \cite{acemoglu2007competition}. Similar price competition models have also been applied in dynamic spectrum sharing when the service providers sell their idle spectrum on the secondary market\cite{kavurmacioglu2012competition} and to study tiered spectrum sharing \cite{liu2014competitionopen}.
	
	In this paper, we consider a scenario similar to that in \cite{nguyen2011impact,nguyen2015free}, where incumbent and entrant SPs compete for customers by announcing prices for their service. The customers select SPs based on the sum of the price they pay for service and a congestion cost that is incurred for using the given band of spectrum. In \cite{nguyen2011impact,nguyen2015free}, the SPs compete by announcing one price for service in an unlicensed band and a different price for service in any licensed band which the SP may own. A main result is that the price of the unlicensed band is always competed down to $0$, meaning that no profit is made for unlicensed service. It is further shown that this may result in a loss of social welfare compared to the case where the unlicensed spectrum is not available.  Additionally, the threat of overcrowding and no-profits may deter SP's from investing and offering service in the first place as studied in\cite{zhou2012investment}. 
	
	Instead of competing on licensed and unlicensed spectrum separately as in \cite{nguyen2011impact,nguyen2015free}, we consider a different model where the incumbent SPs {\it bundle} their licensed and unlicensed service and only announce a single price for this combined service. Customers that accept the price  are served on both licensed and unlicensed bands, where the specific band used may vary over time depending on the customer locations and network control decisions. For such a bundled service, we model customers as being sensitive to the average congestion they experience across the two spectrum bands. 
	
	Bundling different goods has been widely used in many areas in order to benefit both the profit of sellers and the welfare of customers. For example, sporting and cultural organizations usually offer season tickets, and  Microsoft and Sony sell their consoles bundled with different games. In recent years, the bundling strategy has also been adopted by some of the wireless SPs. One example is the data plans from AT\&T which include unlimited usage on the entire national Wi-Fi hot spot network \cite{Att}. Also, T-Mobile is launching LTE-U in spring 2017 \cite{Tmobile} to serve subscribed customers with both licensed and unlicensed spectrum. The use of bundling in markets has been studied in the economic literature including \cite{adams1976commodity,stremersch2002strategic}. Our model is different from these in that the SPs are bundling congestible resources instead of commodities that might be unpopular in the  market. The SPs have to consider the congestion caused by increasing the number of customers.

	We use an $\alpha$ to denote the average percentage of time that customers are using the unlicensed spectrum when choosing a bundled service. Two scenarios are considered in this paper. In the first part, we consider $\alpha$ to be exogenously determined by user behavior. For example, this may depend on user mobility patterns. In this scenario, we first consider the case where an incumbent (with licensed spectrum) is competing against a single entrant (without licensed spectrum). We show that if $\alpha$ is small enough, bundling will increase the incumbent's profit and, more surprisingly, also the profit of the entrant.  In particular, compared to the case in \cite{nguyen2011impact,nguyen2015free},  bundling can provide both SPs the incentive to enter the unlicensed market. Moreover, we show that for certain values of $\alpha$, bundling can also lead to an improvement in social welfare compared to both the unbundled case and the case where the unlicensed spectrum is instead  licensed to the entrant. Variations of this setting with multiple entrants and multiple incumbents are also considered.

	The second scenario we consider is one in which an SP offering bundled service is able to choose their own $\alpha$ by changing the band in which customers are served over time. This can model, for example, the case where LTE-U is used for the unlicesned band and LTE for the licensed band and the common network infrastructure assigns users dynamically to one band or the other. In this case, we show that there may be a difference in the choice of $\alpha$ which maximizes the SP's profit and that which maximizes social welfare. We characterize this gap in the limit of a large number of symmetric incumbents.

	The rest of the paper is organized as follows. Our model is described in Section \ref{sec:model}.  We first treat $\alpha$ as a fixed parameter depending on customer behavior in Section \ref{sec:fix_alpha} and compare it with the model in  \cite{nguyen2011impact,nguyen2015free}. In Section \ref{sec:varying_alpha}, we view $\alpha$ as a controllable variable that the SPs can choose to maximize either its own profit or the social welfare. Numerical results are shown in Section \ref{sec:numerical_results}. Finally, we conclude in Section \ref{sec:conclusion}. Several proofs are omitted due to space considerations and can be found in the appendix.

	\section{System model}
	\label{sec:model}
	We consider a market with $M$ incumbent SPs and $N$ entrant SPs. The sets of incumbent and entrant SPs  are denoted by  $\mathbb{I}$ and $\mathbb{E}$, respectively. Each incumbent $i\in\mathbb{I}$ is assumed to  possess its own licensed band of spectrum with bandwidth $B_i$, while entrants have no licensed spectrum. There is also a single unlicensed band with bandwidth $W$ that can be used by both the incumbent and entrant SPs.  
	
	The SPs are assumed to compete for a common pool of infinitesimal customers by setting prices for their services. The price announced by SP $i$ is denoted by $p_i$. The SPs then serve all customers that accept their price.  The profit of SP $i$ is then $x_ip_i$ where $x_i$ is the customer mass that accept price $p_i$. In the case of bundling, the customers that are choosing the incumbent SP are allowed to use both the licensed and unlicensed spectrum by paying the accepted price. 
	
	Each SP's service  is characterized by a congestion cost. The congestion that the customers experience in a band is denoted by $g(X,Y)$, which is assumed to be increasing in the total customer mass $X$ on the band and decreasing in the bandwidth $Y$. Here, we assume a specific form $g(\frac{X}{Y})$, where $g(\cdot)$ is  a convex increasing function with $g(0)=0$. We assume customers who choose the bundle service  use the unlicensed band with an average percentage of time $\alpha\in[0,1]$.  Since the incumbent SPs are bundling their service, the congestion that the customers experience when choosing the incumbent SP $i\in\mathbb{I} $ is then the expected congestion on both bands, i.e., $(1-\alpha)g(\frac{(1-\alpha)x_i}{B_i})+\alpha g\left(\frac{\sum\limits_{j\in\mathbb{I}}\alpha x_j  +\sum\limits_{k\in\mathbb{E}}x_k  }{W}\right)$. The congestion experienced by customers who choose the entrant SPs is $g\left(\frac{\sum\limits_{j\in\mathbb{I}}\alpha x_j  +\sum\limits_{k\in\mathbb{E}}x_k  }{W}\right)$. Note as in \cite{nguyen2011impact}\cite{nguyen2015free}, the congestion experienced in the unlicensed band depends on the total traffic assign to that band by any SP, modeling the shared nature of this band. 
	
	The customers select a single SP from whom  to receive service. Given the announced prices by all the SPs, the customers will choose the SP with the lowest  {\em delivered price}, which is the sum of the announced price and the congestion associated with the SP. Hence, for an incumbent SP $i\in\mathbb{I}$, the delivered price $d_i(p_i,{\bf x})$ is denoted by $p_i+(1-\alpha)g(\frac{(1-\alpha)x_i}{B_i})+\alpha g\left(\frac{\sum\limits_{j\in\mathbb{I}}\alpha x_j +\sum\limits_{k\in\mathbb{E}}x_k  }{W}\right)$ , where $\mathbf{x}$ denotes the vector of $x_i$'s across the SPs. For an entrant SP $i\in\mathbb{E}$, its delivered price is given by $d_i(p_i,{\bf x}) =  p_i + g\left(\frac{\sum\limits_{j\in\mathbb{I}}\alpha x_j +\sum\limits_{k\in\mathbb{E}}x_k  }{W}\right)$. 
	
	We assume that customers are characterized by an inverse demand function $P(q)$, which means up to a mass of $q$ customers are willing to pay the delivered price $P(q)$. We assume the inverse demand function $P(q)$ is concave decreasing. Each customer is infinitesimal so that a single customer has a negligible effect on the congestion in any band. Therefore, given the announced price by the SPs, the demand of service for each SP $i$ is assumed to satisfy the Wardrop equilibrium conditions \cite{wardrop1952some}. In our model, the conditions for the SPs are
	\begin{eqnarray}
	\label{eqn:wardrop1}
	&&d_i(p_i,{\bf x}) = P\left(\sum\limits_{j\in\mathbb{I}\cup\mathbb{E}} x_j\right),{\rm for}  \;  x_i>0,\nonumber\\
	&&d_i(p_i,{\bf x}) \ge P\left(\sum\limits_{j\in\mathbb{I}\cup\mathbb{E}} x_j\right),{\rm for} \;\forall i. 
	\end{eqnarray}

	The conditions in (\ref{eqn:wardrop1}) imply that at the Wardrop equilibrium, all the SPs serving a positive amount of customers will end up with the same delivered price, which is given by the inverse demand function. A (Nash) equilibrium of the  game is one in which the customers are in a Wardrop equilibrium and no SP can improve their profit by changing their announced price. 
	
	At an equilibrium, the customer surplus is defined as the difference between the delivered price each customer pays and the amount it is willing to pay, integrated over all the customers, i.e.,  
	\begin{equation}
	\label{eqn:csdefinition}
	CS = \int_{0}^{Q}P(q)-P(Q)dq,
	\end{equation}
	where $Q = \sum\limits_j x_j$. The social welfare of the market is the sum of consumer welfare and the SPs' profits:
	\begin{equation}
	\label{eqn:sedefinition}
	SW = CS+ \sum\limits_j p_jx_j.
	\end{equation}
	
	Next we briefly introduce the models that we want to compare with the bundling case.
	\subsection{Price competition with unbundled unlicensed access}
	\label{sec:unlicensed}
	The first model we want to contrast with is price competition with unlicensed access modeled in \cite{nguyen2011impact,nguyen2015free}. In this model, the incumbent SPs will announce two prices for service on the license band and  unlicensed band separately. Customers will still choose the service which offers the lowest delivered price. 
	
	Precisely, for an incumbent SP $i\in\mathbb{I}$, the problem is to maximize its own profit $p^l_ix^l_i+p^u_ix^u_i$, where  the superscripts $l$ and $u$ denote the licensed and unlicensed band, respectively. For an entrant SP $i\in\mathbb{E}$, the problem is to maximize $p^u_ix^u_i$. Again the announced price and the resulting customer mass on each band should meet the Wardrop equilibrium conditions. For the licensed band of each SP $i$ these are 
	\begin{small}
		\begin{eqnarray}
		& p^l_i+g\left(\frac{x^l_i}{B_i}\right)= P\left(\sum\limits_{j\in\mathbb{I}} x^l_j+\sum\limits_{j\in\mathbb{I}\cup\mathbb{E}} x^u_j\right),{\rm for} \; i\in\mathbb{I}, x^l_i>0\nonumber\\
		\label{eqn:unlicense_wardrop1}
		& p^l_i+g\left(\frac{x^l_i}{B_i}\right)\ge P\left(\sum\limits_{j\in\mathbb{I}} x^l_j+\sum\limits_{j\in\mathbb{I}\cup\mathbb{E}} x^u_j\right),{\rm for} \; i\in\mathbb{I}.
		\end{eqnarray}
	\end{small}
	For the unlicensed band, the Wardrop equilibrium conditions are 
	\begin{small}
		\begin{eqnarray}
		& p^u_i +g\left(\frac{ \sum\limits_{j\in\mathbb{I}\cup\mathbb{E}} x^u_j}{W}\right) = P\left(\sum\limits_{j\in\mathbb{I}} x^l_j+\sum\limits_{j\in\mathbb{I}\cup\mathbb{E}} x^u_j\right),\nonumber\\
		&{\rm for} \; i\in\mathbb{I}\cup\mathbb{E},x^u_i>0, \nonumber\\
		\label{eqn:unlicense_wardrop2}
		& p^u_i +g\left(\frac{ \sum\limits_{j\in\mathbb{I}\cup\mathbb{E}} x^u_j}{W}\right) \ge P\left(\sum\limits_{j\in\mathbb{I}} x^l_j+\sum\limits_{j\in\mathbb{I}\cup\mathbb{E}} x^u_j\right),{\rm for} \; i\in\mathbb{I}\cup\mathbb{E}.
		\end{eqnarray}
	\end{small}

	
	We summarize the Nash equilibrium for the unlicensed access model in \cite{nguyen2011impact} in the following lemma.
	\begin{lemma}
		\label{lemma:price}
		In the unbundled unlicensed access case, if there are at least two SPs in the market, all the SPs serving a positive mass of customers will have $p^u_i = 0$ on the unlicensed band. 
	\end{lemma}
	The lemma shows that under this competition model,  the entrant SP  always gets no  profit.
	
	\subsection{Exclusive use}
	\label{sec:exclusive}
	An alternative to making the spectrum with bandwidth $W$ unlicensed and still increase competition would be to licensed this band to a new entrant. \footnote{Of course this band could also be licensed to an incumbent, but this would not increase competition in the market and so we focus on the case where it is licensed to an entrant.}  In this case, the entrant SP is able to use the spectrum exclusively and incumbent SPs are only allowed to use their own proprietary band. The objective of each SP is still to maximize profit  while the Wardrop equilibrium conditions are satisfied. To be precise, the Wardrop equilibrium conditions on the licensed band are the same as that in (\ref{eqn:unlicense_wardrop1}). The conditions for the entrant SP $i$ on the new band become
	\begin{eqnarray}
	& p_i+g\left(\frac{x_i}{W}\right)= P\left(\sum\limits_{j\in\mathbb{I}\cup\mathbb{E}} x_j\right), {\rm if}\;\; x_i>0\nonumber\\
	\label{eqn:exclusive_wardrop1}
	& p_i+g\left(\frac{x_i}{W}\right)\ge P\left(\sum\limits_{j\in\mathbb{I}\cup\mathbb{E}} x_j\right),{\rm otherwise} .
	\end{eqnarray}
	Note that here we ignore the superscript because the spectrum can only be used by the entrant SP. 
	
	\section{Different cases with fixed $\alpha$}
	\label{sec:fix_alpha}
	\subsection{Monopoly case}
	We first consider the case with only one incumbent SP and no entrant SPs. The game then reduces to the following optimization problem for the incumbent SP:
	\begin{eqnarray}
	\label{eqn:monobundle}
	\max_{p_1}&&p_1x_1\\
	{\rm s.t.} && d_1(p_1,x_1) = P(x_1),\nonumber\\
	&&  p_1\ge0.\, \nonumber
	\end{eqnarray}
	where  $d_1(p_1,x_1)$ is the delivered price.
	
	We want to compare the bundling case with the unbundled and exclusive use cases described in previous section. When there is only one incumbent in the market, the unbundled and exclusive use cases coincide with each other and result in the following optimization problem:
	\begin{eqnarray}
	\label{eqn:monoprice}
	\max_{p^l_1,p^u_1}&&p^l_1x^l_1+p^u_1x^u_1\\
	{\rm s.t.} && p^l_1+g\left(\frac{x^l_1}{B}\right)= P(x^l_1+x^u_1),\nonumber\\
	&& p^u_1 +g\left(\frac{ x^u_1}{W}\right) = P(x^l_1+x^u_1),\nonumber\\
	&&  p^l_1\ge0, p^u_1\ge0. \nonumber
	\end{eqnarray}
	\begin{theorem}
		\label{thm:mono}
		The monopoly SP cannot gain more profit by bundling the licensed and unlicensed service than by offering service  on the two bands separately.
	\end{theorem}

	
	Intuitively, in a monopoly market, there are no competitors for the incumbent SP on the unlicensed band, so the incumbent SP can get as much as profit as possible on both bands. As a result there are no incentives for the SP to bundle the services. The proof of this theorem  shows that when there is only one incumbent SP in the market, there exists an optimal $\alpha^*$ such that when $\alpha = \alpha^*$, the SP is indifferent between bundling and unlicensed access. When $\alpha\ne \alpha^*$, the incumbent SP will generally get a lower profit by bundling as is shown in Fig. \ref{fig:monopoly}, where the solid curve reflects the profit in the bundling case and the dashed line indicates the unbundled case. In any settings, the profit of the SP in the bundling case cannot exceed the profit in the unlicensed access case. With different licensed and unlicensed bandwidth ratios $(B/W)$, $\alpha^*$ varies. In fact, it can be proved that $\alpha^* = \frac{W}{B+W}$.
	
	\begin{figure}[htbp]
		\centering
		\includegraphics[scale=0.35]{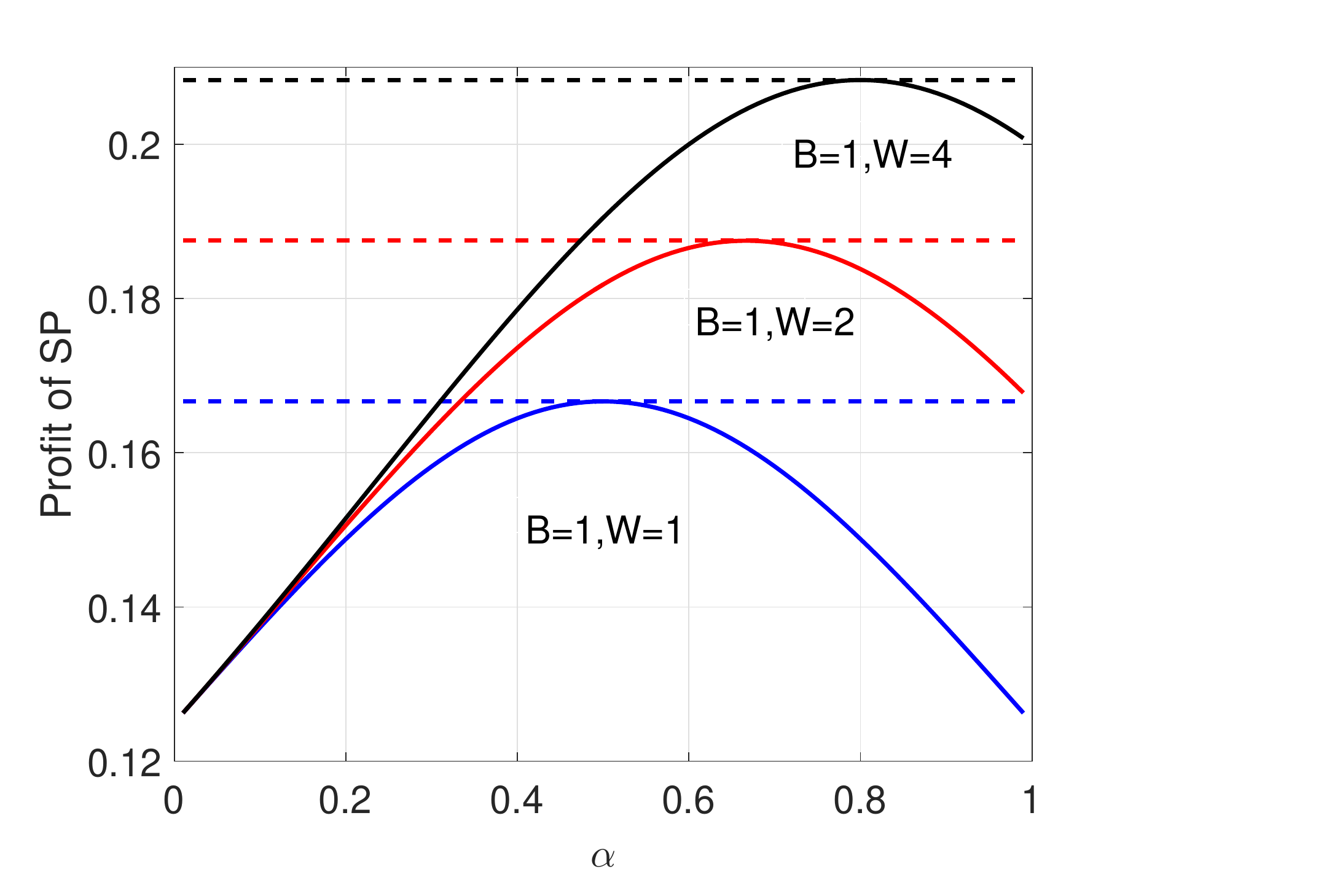}
		\caption{The profit of the monopoly SP with different $\alpha$  (dashed lines show the unbundled case, solid lines the bundled case).}
		\label{fig:monopoly}
	\end{figure}

	\subsection{One incumbent SP \& one entrant SP }
	Next we turn to the case of one incumbent and one entrant SP. Different from the monopoly case, when there is an entrant SP without any proprietary spectrum, the incumbent SP is able to make higher profit by bundling the licensed and unlicensed service. Without loss of generality, we assume that SP 1 is the incumbent and SP 2 is the entrant. 
	
	\subsubsection{Comparison with unbundled case}
	First, we compare bundling to the case without bundling.  Based on Lemma \ref{lemma:price}, the announced price on the unlicensed band is zero in the unbundled case and as a result the incumbent SP is solving the following optimization problem:
	\begin{eqnarray}
	\label{eqn:priceoptimization}
	\max_{p_1} && p_1 x^l\\
	{\rm s.t.} && p_1 +g\left(\frac{x^l}{B}\right) = P\left(x^l+x^u\right),\nonumber\\
	&& g\left(\frac{x^u}{W}\right) = P\left(x^l+x^u\right), \nonumber\\
	&& p_1\ge 0 .\nonumber
	\end{eqnarray}
	
	The following result will aid in comparing this with the bundled case.
	\begin{theorem}
		\label{thm:incumbent_entrant1}
		There always exists an $\alpha_0\in(0,1)$, such that when $\alpha = \alpha_0$, bundling is equivalent to the unbundled case in the sense of profit and welfare. 
	\end{theorem}
	
	Theorem \ref{thm:incumbent_entrant1} shows that bundling is a kind of generalization of the unbundled case. It can achieve exactly the same profit for the SPs and welfare for the market with some value of $\alpha$. In that case, the incumbent SP actually lowers its announced price (compared to the case where $W$ is not present) and customers can use both the licensed and unlicensed band. No customer would like to choose the entrant even if it offers a price of $0$, because the customers who are choosing the bundled service have already generated significant congestion on the unlicensed band. However when $\alpha$ is different from the critical value it is possible for entrant SP to gain positive profit and for the incumbent to increase its profit. 
	
	\begin{theorem}
		\label{thm:incumbent_entrant2}
		Let $\alpha_0$ be the value at which bundling achieves the same profit and welfare as the unbundled case. When $\alpha<\alpha_0$, the profit of incumbent SP in the bundling case is higher than that in the unbundled case and the entrant SP can gain positive profit. When $\alpha>\alpha_0$, the profit of the incumbent SP in the bundling case is no greater than that in the unlicensed case and the entrant SP gets zero profit.
	\end{theorem}
	
	This shows that when $\alpha$ is small, both the incumbent and entrant SPs are serving customers with a positive announced price. The unlicensed band can offload some congestion on the licensed band so that the incumbent SP can attract more customers and make more profit. In this case, the unlicensed band is not over-crowded so that the entrant SP can still announce a positive price to gain positive profit. However when $\alpha$ is large, the entrant SP is forced to announce a zero price, because the customers of the incumbent SP have already caused significant congestion on the unlicensed band. That means the incumbent cannot gain profit on unlicensed band, but it still has to serve a certain amount of customers on it. That may have a negative effect on his profit. 
	
	\subsubsection{Comparison with exclusive use case}
	From Lemma \ref{lemma:price} it is clear that  price competition for unbundled unlicensed  access will benefit the customers instead of the SPs, since the competition on the unlicensed band will drive the price to $0$. Another way to use the additional spectrum is the exclusive use case. As described in previous section, exclusive use means that the entrant SP can use the unlicensed spectrum exclusively. It is natural to ask if bundling can gain more profit for the incumbent SP than the exclusive use case. First we give a lemma to compare the profit of incumbent in the exclusive case and unlicensed case.
	\begin{lemma}
		\label{lemma:exclusive_unlicense}
		In the case with one incumbent and one entrant SP, the incumbent SP is able to gain more profit with exclusive access than with unbundled unlicensed  access.
	\end{lemma}
	\begin{proof}
		Note that the exclusive use case is equivalent to the bundling case with $\alpha = 0$. Hence by Theorem \ref{thm:incumbent_entrant2}, the incumbent can always gain more profit in the exclusive use case than the unbundled case. 	
	\end{proof}
	Surprisingly, this lemma show that without bundling, rather than making the spectrum unlicensed, the incumbent would prefer that it is exclusively licensed to the entrant.  However, things may become different if bundling is adopted by the incumbent.  To give more insights, we next consider a setting where both the inverse demand and the congestion function have the linear forms  $P(x)=1-x$ and $g(x) = x$.
	
	\begin{theorem}
		\label{thm:exclusiveprofit}
		For linear case, if $B<\frac{4(1+W)}{3W}$, there must exists some $\alpha \in (0,1)$ such that the profit of the incumbent SP is higher than that in the exclusive use case.
	\end{theorem}
	
	This theorem shows that the incumbent SP  can only gain more profit in the bundling case if the licensed resources it possesses are limited. Another interpretation of Theorem \ref{thm:exclusiveprofit} is that when the licensed band is small enough, i.e., $B<\frac{4}{3}$, for all possible value of $W$, there exists some $\alpha$ such that  bundling can gain more profit for incumbent than the exclusive use case. However if $B>\frac{4}{3}$, then we require a small amount of unlicensed bandwidth, i.e., $W<\frac{4}{3B-4}$ to guarantee the existence of an $\alpha$ such that bundling is better for the incumbent SP.  
	
	Lemma \ref{lemma:exclusive_unlicense} shows that incumbent SP can gain more profit with exclusive access than that in the unbundled case, so there is no incentive for the incumbent to share the unlicensed spectrum with the entrant. However, Theorem \ref{thm:exclusiveprofit} shows that  if the incumbent SP uses the bundling strategy, it is able to get even more profit than the exclusive use case if $\alpha$ is in some appropriate range.  Fig. \ref{fig:exclusive_profit}(a) shows how the profit of the incumbent changes with $\alpha$ in linear case with $B=1$ and $W=1$. It can be seen that for some $\alpha$ bundling gives a better profit for the incumbent. However when $B$ is large, as is shown in Fig. \ref{fig:exclusive_profit}(b) with $B=3$ and $W=1$, the profit of the exclusive use case is always higher. In this case, the incumbent SP has sufficient licensed spectrum so to not be incentivized to use the unlicensed band.
	
	\begin{figure}[htbp]
		\centering
		\subfigure[$B=1,W=1$]{
			\label{fig:exclusive_profit:a} 
			\includegraphics[width=1.68in]{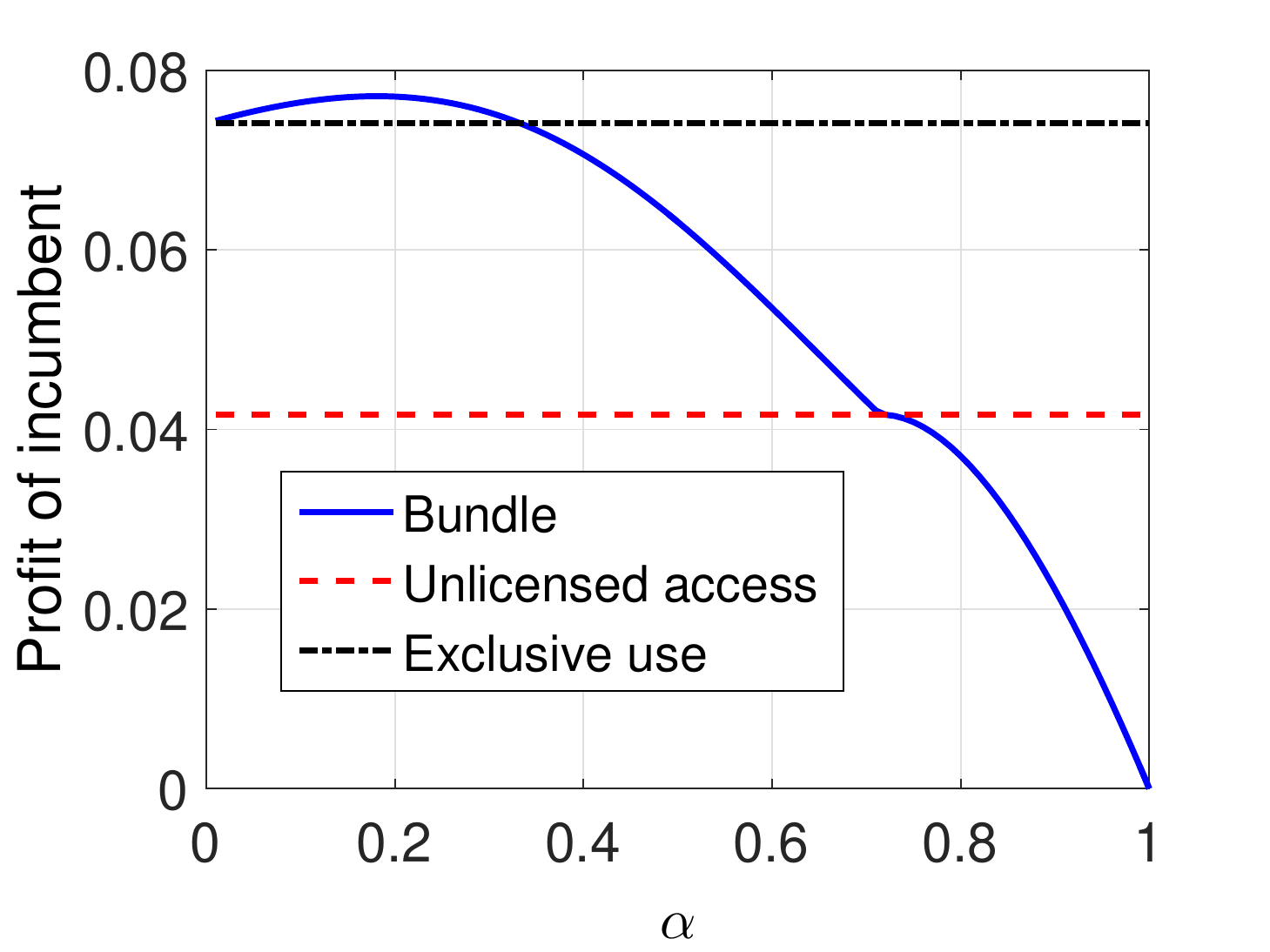}}
		\subfigure[$B=3,W=1$]{
			\label{fig:exclusive_profit:b} 
			\includegraphics[width=1.68in]{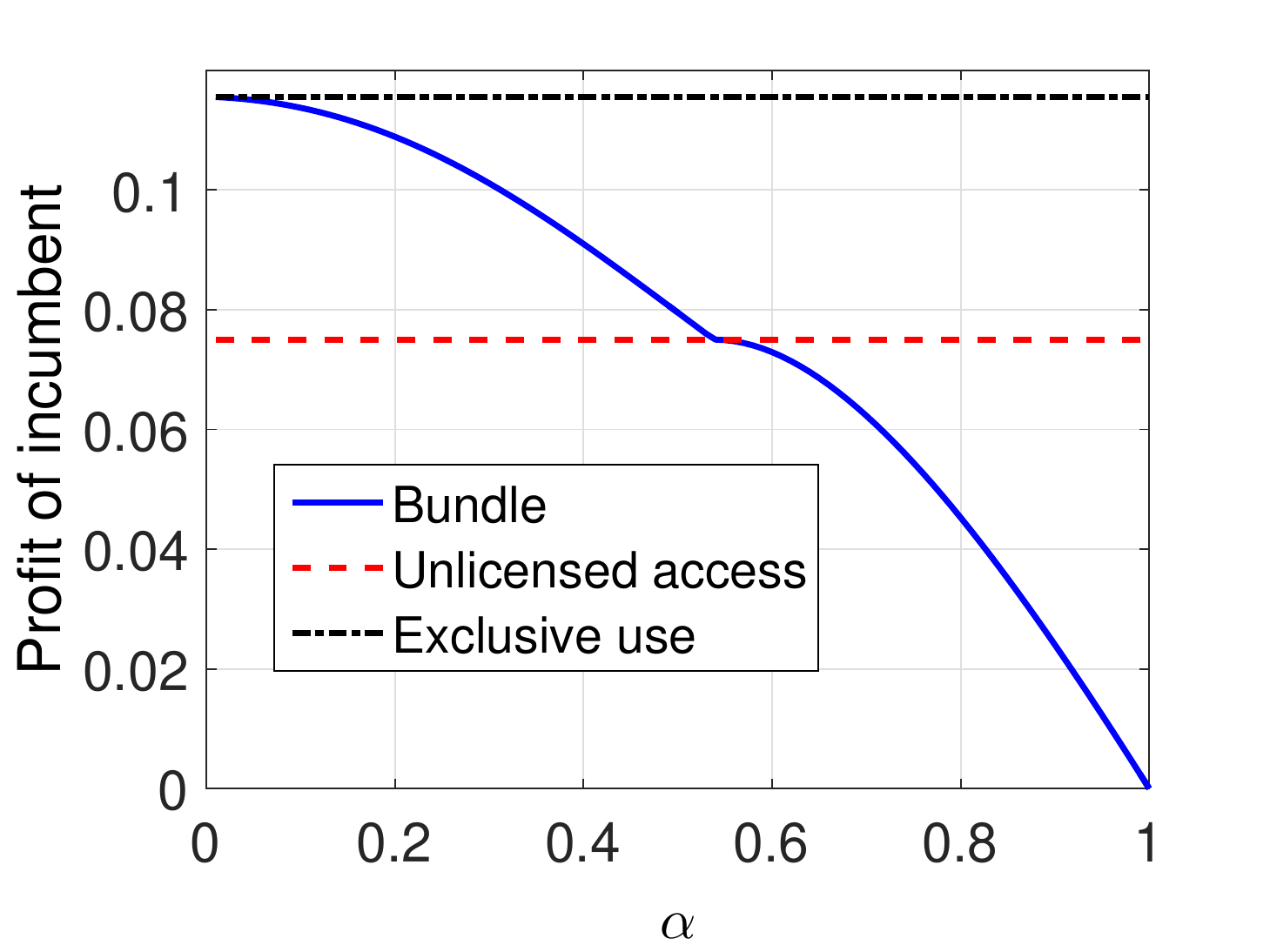}}
		\caption{Profit comparison with unbundled and exclusive use cases.}
		\label{fig:exclusive_profit} 
	\end{figure}
	%
	
	Next we show that the social welfare may also be improved over the exclusive use case if bundling is adopted.
	\begin{theorem}
		\label{thm:exclusiveSW}
		For any value of $B$ and $W$, there exists some $\alpha>0$ such that both of the customer welfare and social welfare of bundling is better than that of the exclusive use case.
	\end{theorem}
	
	Different from the profit of incumbent, it is always possible for customer surplus and social welfare to be improved if the incumbent chooses to bundle licensed and unlicensed service for any $B$ and $W$. Two examples are shown in Fig. \ref{fig:exclusive_sw}. In both cases we can find some $\alpha$ such that social welfare increases. That is because when $\alpha$ is in some specific range, the customers who choose the incumbent SP are able to use the unlicensed service, which introduces competition on the unlicensed band and as a result improves the social welfare. However when $\alpha$ is too large, the licensed band becomes underutilized, and customer and social welfare may suffer. 
	\begin{figure}[htbp]
		\centering
		\subfigure[$B=1,W=1$]{
			\label{fig:exclusive_sw:a} 
			\includegraphics[width=1.68in]{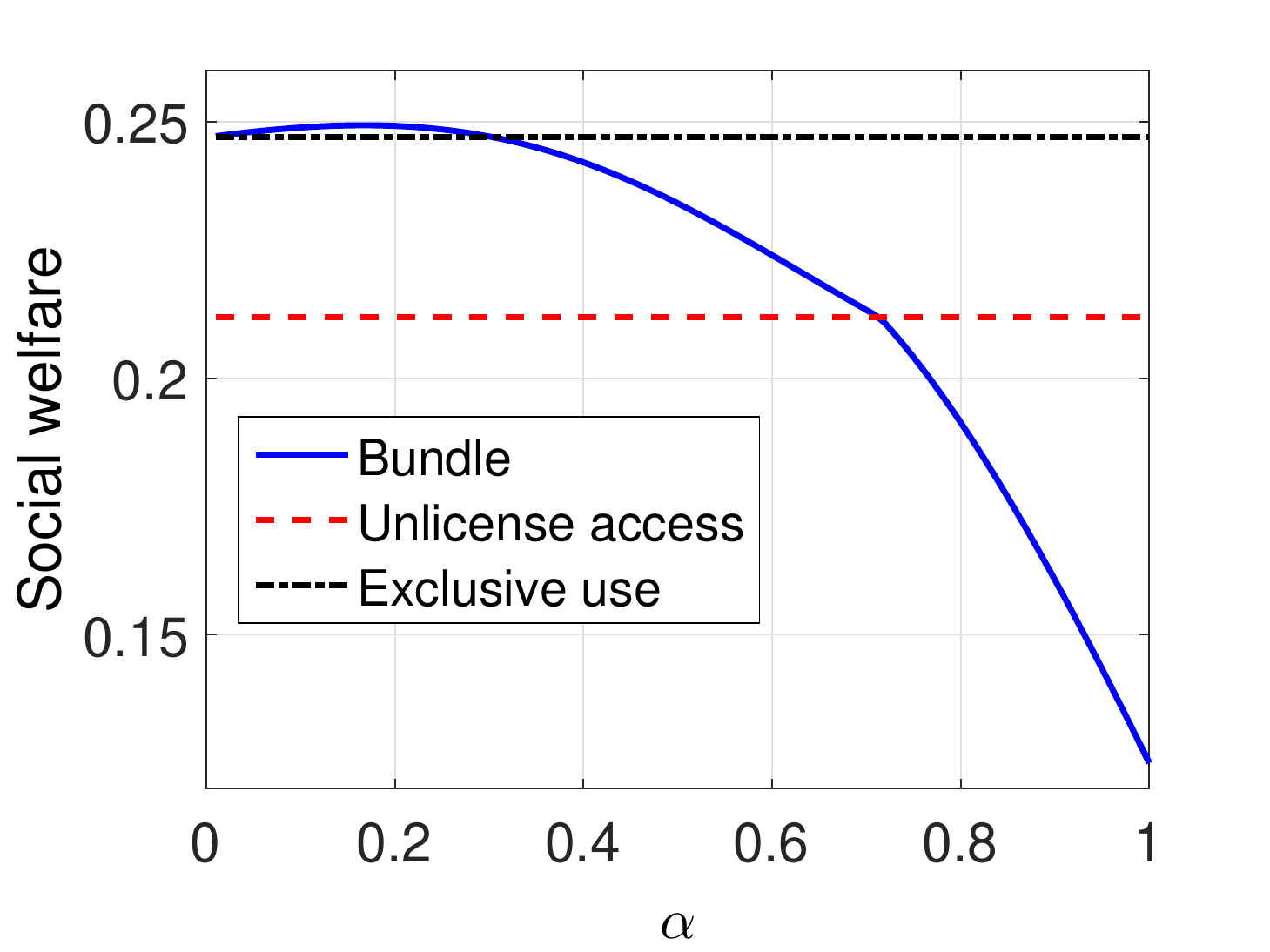}}
		\subfigure[$B=3,W=1$]{
			\label{fig:exclusive_sw:b} 
			\includegraphics[width=1.68in]{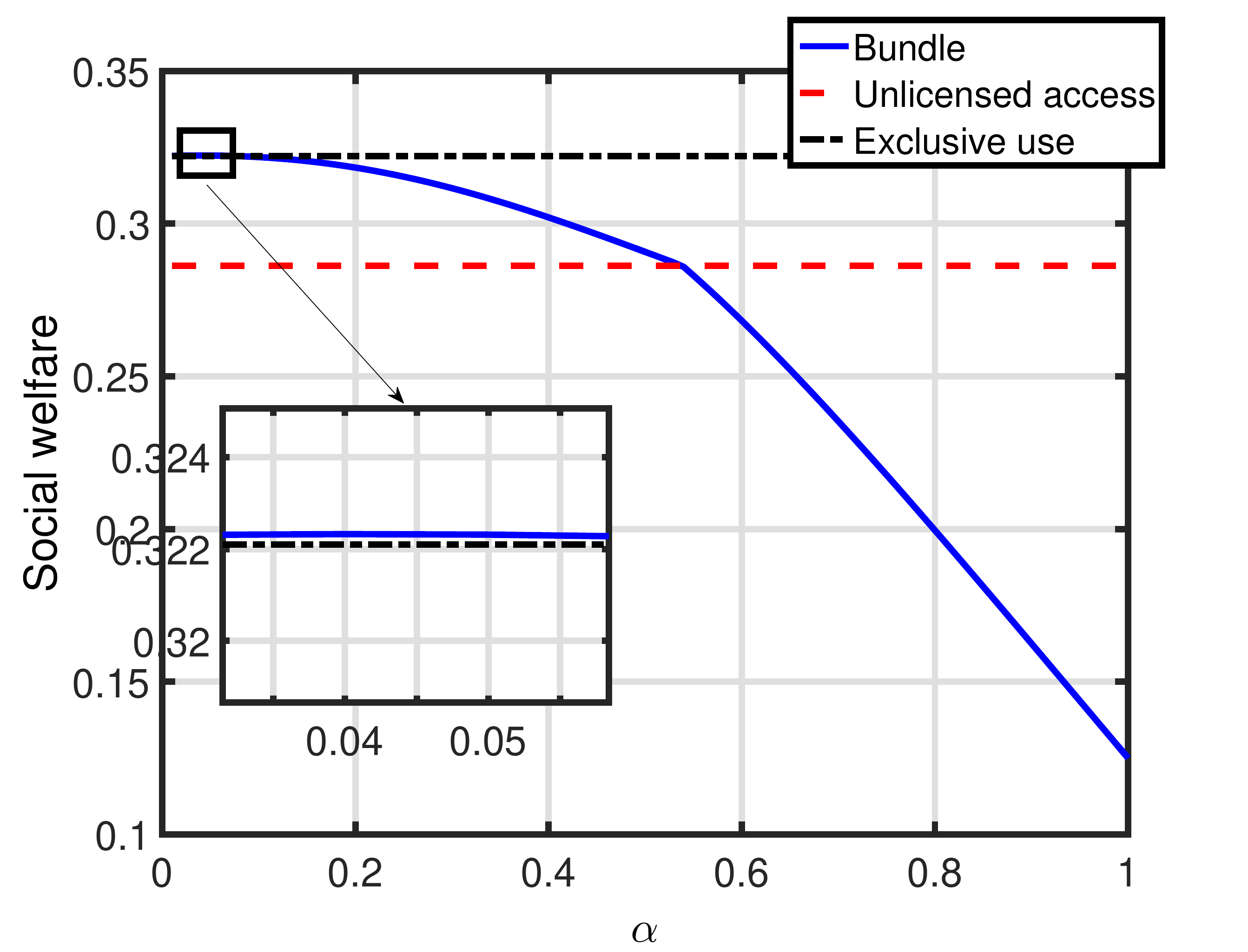}}
		\caption{Social welfare comparison with unbundled and exclusive use cases when $W$ is small.}
		\label{fig:exclusive_sw} 
	\end{figure}
	
	
	Note that in Fig. \ref{fig:exclusive_sw}, we show the social welfare with a relatively small bandwidth of unlicensed spectrum. In this case, assigning the entrant SP exclusive access to the unlicensed spectrum yield more social welfare than unbundled access. However, when $W$ is large, it is possible for the unbundled case to yield more social welfare than the exclusive use case. Such cases are shown in Fig \ref{fig:exclusive_sw4}. We can see that when $B=1$ and $W=10$, social welfare in the unbundled case is higher than that in the exclusive use case and bundling is able to achieve higher social welfare in some range of $\alpha$. Even if $\alpha \to\infty$, where the unbundled case achieves the maximum possible social welfare, there exists an $\alpha$ such that bundling can achieve the same social welfare. That is because there always exists some $\alpha$ such that the bundling case is equivalent to the unbundled access case  as is shown in Theorem \ref{thm:incumbent_entrant1}.
	\begin{figure}[htbp]
		\centering
		\subfigure[$B=1,W=10$]{
			\label{fig:exclusive_sw4:a} 
			\includegraphics[width=1.68in]{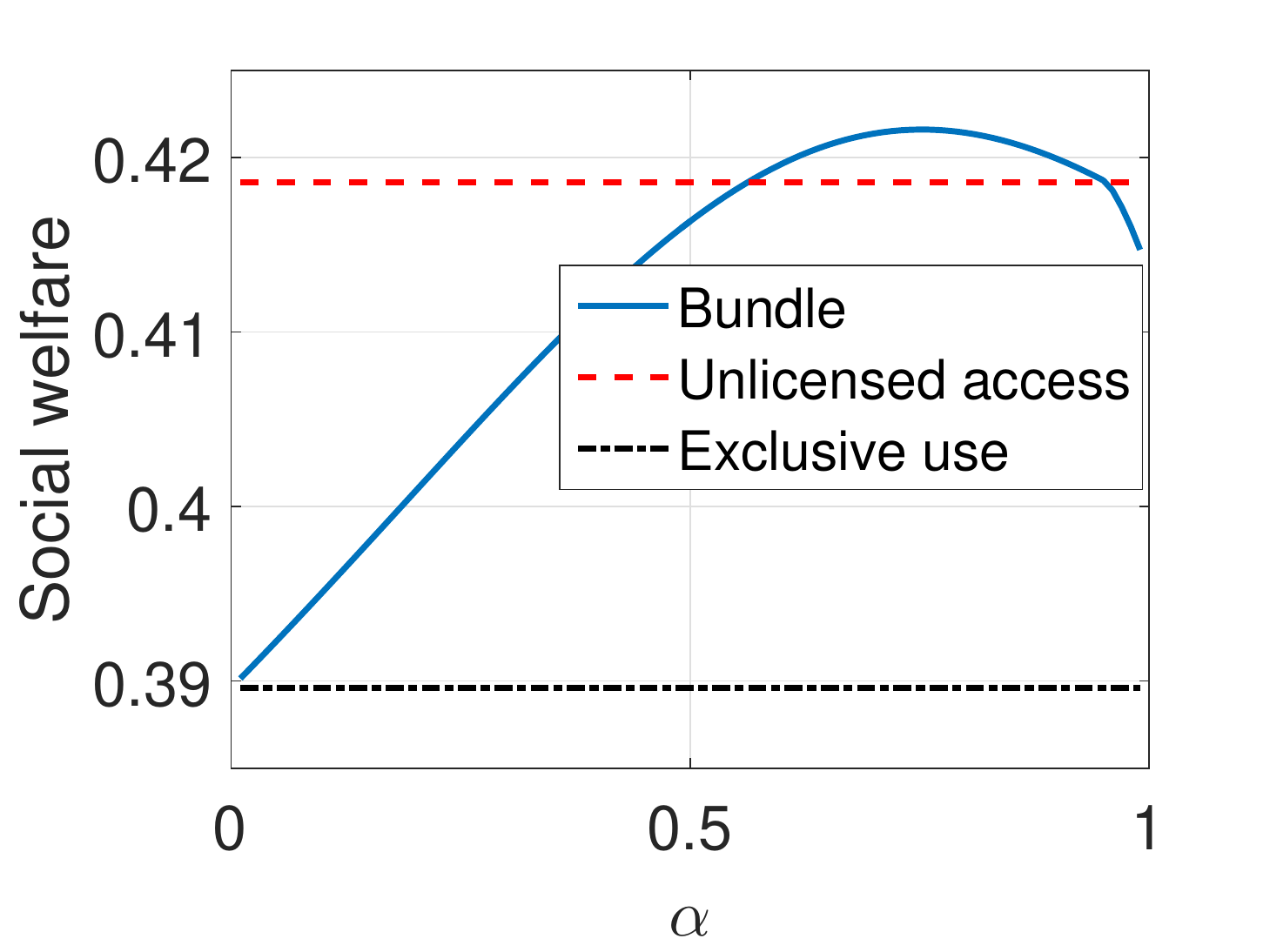}}
		\subfigure[$B=1,W\to\infty$]{
			\label{fig:exclusive_sw4:b} 
			\includegraphics[width=1.68in]{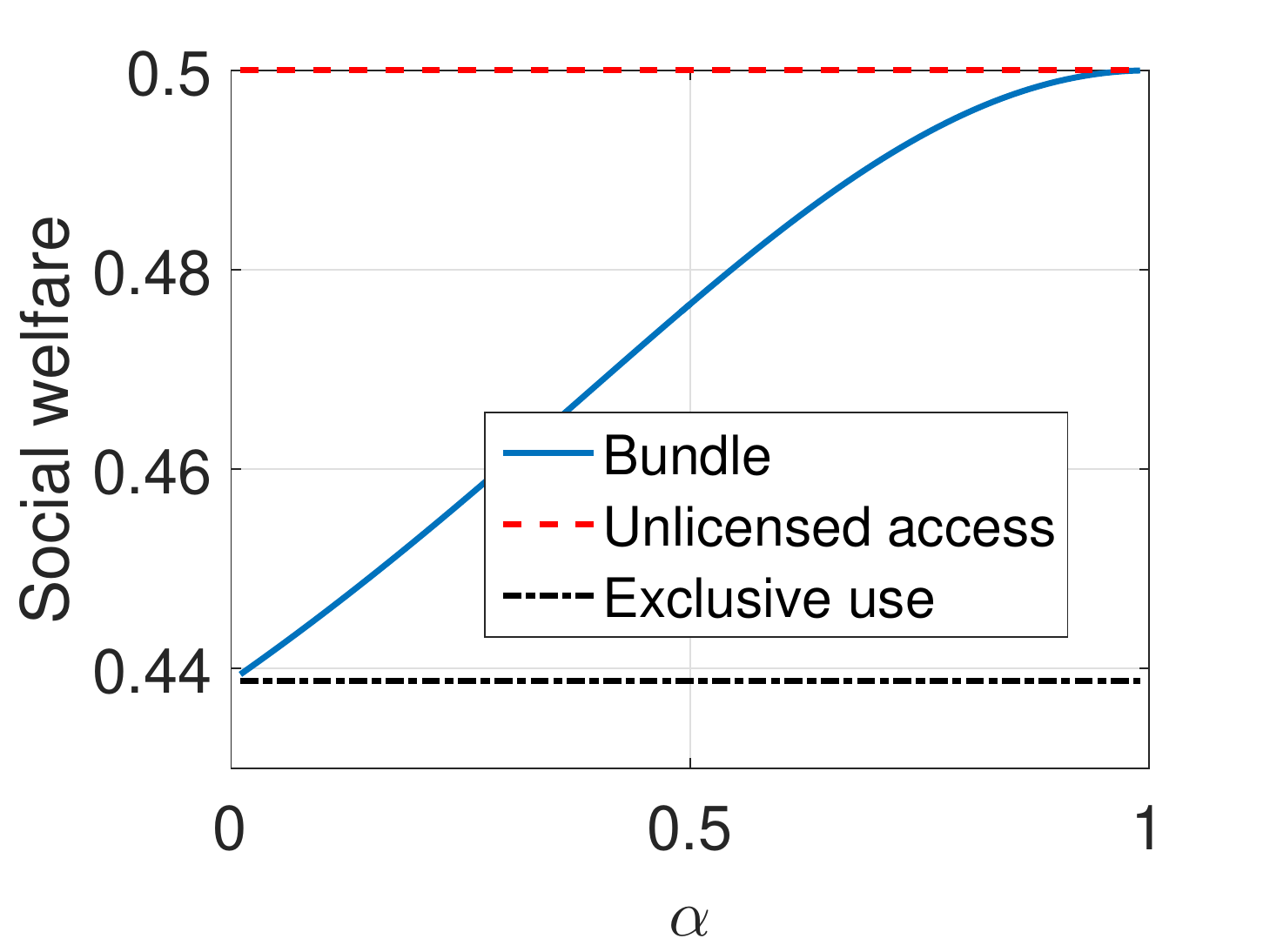}}
		\caption{Social welfare comparison with unbundled and exclusive use cases when $W$ is large.}
		\label{fig:exclusive_sw4} 
	\end{figure}


	\subsection{One incumbent SP \& multiple entrant SPs}
	In this case, the competition among entrants on the unlicensed spectrum will drive the price to $0$, which results in an equilibrium similar to  the unbundled case. We summarize the result in the following theorem.  
	\begin{theorem}
		\label{thm:incumbent_multipleentrant1}
		Where there is one incumbent SP and multiple entrant SPs, for any $\alpha \in (0,1)$, bundling is equivalent to the unbundled case in the sense of profit and welfare.
	\end{theorem}
	\begin{proof}
		Because there are multiple entrant SPs in the market, based on Lemma \ref{lemma:price}, the announced price of the entrant SPs should be $0$. That reduces the game to the following optimization problem: 
		\begin{eqnarray}
		\label{eqn:multi_entrant}
		\max_{p_1}&&p_1x_1\\
		{\rm s.t.} && d_1(p_1,\mathbf{x}) = P(x_1+\sum\limits_{j\in\mathbb{E}}x_j), \nonumber\\
		&& d_i(0,\mathbf{x}) = P(x_1+ \sum\limits_{j\in\mathbb{E}}x_j), \; i\in\mathbb{E} \nonumber\\
		&&  p_1\ge 0, \nonumber
		\end{eqnarray}
		where $d_1(p_1,\mathbf{x})$ is the delivered price of the incumbent SP and $d_i(0,\mathbf{x}) $  is the delivered price of the entrant SPs. 
		

		Let $x^l = (1-\alpha)x_1$ and $x^u = \alpha x_1 + \sum\limits_{j\in\mathbb{E}}x_j$, the optimization problem can be transformed to the following form:
		\begin{eqnarray}
		\label{eqn:multientrant}
		\max_{p_1}&&\frac{p_1}{1-\alpha}x^l\\
		{\rm s.t.} && p_1+g\left(\frac{x^l}{B}\right)= P(x^l+x^u),\nonumber\\
		&& g\left(\frac{ x^u}{W}\right) = P(x^l+x^u),\nonumber\\
		&&  p_1\ge0. \nonumber
		\end{eqnarray}
		The problem in (\ref{eqn:multientrant}) is exactly the same as the problem of unbundled access in (\ref{eqn:priceoptimization}). This implies identical profit for the incumbent and identical welfare for the market.
	\end{proof}
	
	Since with multiple entrants, the price for unlicensed service goes to zero, this results suggests that if the investment decisions of the entrants are also accounted for as in \cite{zhou2012investment} then more than one entrant would never enter the market, i.e., cases with either one or zero entrants are more relevant. This result can be simply extended to the cases with multiple incumbent SPs and multiple entrant SPs. When there are multiple entrant SPs, the competition on the unlicensed band will again drive the price to $0$. 
	
	
	\subsection{Multiple incumbents \& no entrants}
	Next we turn to the case with multiple incumbents and no entrants.\footnote{As discussed in the previous section, given that incumbents can always make a higher profits, it is reasonable that with a high enough entry cost that no entrants will enter the market.}  In this case, the incumbent SPs will compete on their price for bundled service. The problem can be formulated as the following for each incumbent SP $i$,
	\begin{eqnarray}
	\label{eqn:Nincumbent_problem}
	\max_{p_i}&&p_ix_i \nonumber\\
	{\rm s.t.} && d_i(p_i,\mathbf{x})= P\left(\sum\limits_j x_j\right),\\
	&& p_i\ge0, i=1,2,..., M, \nonumber
	\end{eqnarray}
	where $d_i(p_i,\mathbf{x})$ is the delivered price of SP $i$. First, we characterize the Nash equilibrium of this game.
	
	\begin{theorem}
		\label{thm:equilibrium_existence}
		There exists a Nash equilibrium for the game with multiple incumbents and no entrants  and bundled service.  If $B_i > B_j$, the equilibrium profit of SP $i$ is higher than that of SP $j$.
	\end{theorem}
	We use Debreu's theorem \cite{debreu1952social} to prove the existence of Nash equilibrium. In the proof, we need to show that the strategy space is compact and convex, the profit of one SP should be continuous in all other SPs's announced prices and concave in its own price. Compactness of strategy space and concavity of the pay-offs are easy to show. We then use the results in \cite{menache2011network} to show the continuity part. Details are omitted here due to space constrains. 
	
	
	\begin{corollary}
		\label{col:symmetric_equilibrium} 
		If each SP has the same amount of licensed spectrum, there exists a symmetric Nash equilibrium. 
	\end{corollary}
	The Corollary follows directly from Theorem \ref{thm:equilibrium_existence} based on Theorem 3 in \cite{cheng2004notes}.

	By making stronger assumptions on the inverse demand and congestion function we get stronger results on the equilibrium. 
	\begin{theorem}
		\label{thm:supermodular}
		Assume the inverse demand function $P(x)$ and congestion cost $g(x)$ are both linear, the game with multiple incumbent SPs and no entrant SP is supermodular.
	\end{theorem}

	Supermodularity guarantees the existence of a Nash equilibrium, which is consistent with the result in Theorem \ref{thm:equilibrium_existence}. Also, supermodularity can give us more insights on the equilibrium.
	
	%

	\begin{theorem}
		\label{thm:linear_alpha}
		In the case with multiple incumbent SPs and no entrant SP, assume the inverse demand function $P(x)$ and congestion cost $g(x)$ are both linear. There exists some $\alpha_0\in(0,1)$, such that when $\alpha<\alpha_0$, the profit of each incumbent SP is higher than in the unbundled  case.
	\end{theorem}
	
	The proof of the theorem uses the supermodularity from Theorem \ref{thm:supermodular}. We first consider an auxiliary case with $M$ SPs and one entrant SP and show that there exists some $\alpha$ such that the profit of incumbent SPs in the bundling case is the same as that in the unbundled case. Then we show that when the single entrant SP is removed from the market, the incumbent SPs will have incentive to increase their announced prices to increase profit. The supermodular property then guarantees that each SP's announced price and profit will change towards the same direction.
	
	Theorem \ref{thm:linear_alpha} shows that, when $\alpha$ is relatively small, applying the bundling scheme may benefit each of the incumbent SPs. This means when the percentage of customers using unlicensed service is low, the SPs are more willing to bundle the service to gain more profit. 
	
	Next we consider the limit of a large amount of unlicensed spectrum and a more general class of congestions functions. 
	\begin{theorem}
		\label{thm:expanding_band}
		In the case with multiple incumbent SPs and no entrant SP, when $W\to \infty$, with congestion function $g(x) = kx^p$, $k>0,p\ge1$, using bundling is equivalent to expanding the licensed band of each SP $i$  from $B_i$ to $\frac{B_i}{(1-\alpha)^{p+1}}$ and having no unlicensed spectrum $W$.	
	\end{theorem}

	Theorem \ref{thm:expanding_band} shows that when the bundling scheme is applied, adding a large amount of unlicensed spectrum to the market actually benefits the incumbent by expanding the licensed band with a factor $\frac{1}{(1-\alpha)^{p+1}}$. On the contrary, if the SPs are competing without bundling, adding a large amount of unlicensed spectrum drives the delivered price to $0$, which implies $0$ profits for the incumbents.  In addition, we can see that when $\alpha$ is large, the factor $\frac{1}{(1-\alpha)^{p+1}}$ will be large as well. However, when $\alpha$ is small, even if the regulator adds an infinite amount of unlicensed spectrum to the market, the impact on the market is relatively small.
	
	\begin{theorem}
		\label{thm:n_incumbent_sw}
		Consider the linear symmetric case with $M$ incumbent SPs. The SP's profit, customer surplus and social welfare all increases with the bandwidth of unlicensed spectrum $W$ if all SPs use bundling.
	\end{theorem}

	Note that when the SPs are competing without bundling, although the customer welfare and social welfare may increase with the bandwidth of unlicensed spectrum\footnote{In fact the social welfare may decrease with $W$ when $W$ is small\cite{nguyen2015free}.}, the profit of the SPs is always decreasing in $W$. However, according to Theorem \ref{thm:n_incumbent_sw}, the profit of the SPs is also able to increase with $W$ as well as the customer surplus and social welfare with bundling.

	\section{Controllable $\alpha$}
	\label{sec:varying_alpha}
	In the previous section, we assumed that $\alpha$ is fixed. In this section, we consider $\alpha$ as a parameter that can be determined by the incumbent SPs to maximize either the profit or the social welfare. This can be achieved for example by the SP's network dynamically controlling which form of access a customer is served by. We consider two possible scenarios, selfish incumbents and altruistic incumbents. If the incumbent SPs are selfish, they care more about their profits when choosing $\alpha$. On the contrary, the altruistic SPs may try to optimize the social welfare by choosing appropriate $\alpha$. In this section, we focus on the analytical results when $W\to \infty$ to simplify the calculation. In Section \ref{sec:numerical_results}, we provide numerical results with finite $W$, which show that the results we get in the asymptotic cases can be translated to cases with finite $W$.
	
	\subsection{One incumbent SP and one entrant SP}
	In this subsection, we focus on the case with one incumbent SP and one entrant SP. We formulate it as a two-stage game. In the first stage, the incumbent SP decides its own $\alpha$ and in the second stage, the SPs compete with each other on price for the customers as in the previous sections. First we consider the case that the incumbent SP chooses $\alpha$ to optimize its profit.
	
	\begin{theorem}
		\label{thm:alpha_one_in_one_en}
		Consider a linear model with one incumbent and one entrant SP with $W\to \infty$. When $B\le \frac{4}{3}$, the optimal $\alpha^* = 1-\frac{\sqrt{3B}}{2}$ and the resulting profit is $\frac{1}{48}$. When $B > \frac{4}{3}$, the optimal $\alpha^* = 0$ and the optimal profit of incumbent is $\frac{B}{(4+3B)^2}$.
	\end{theorem}

	The result shows that when there are unlimited unlicensed resources and the bandwidth $B$ of licensed spectrum is relatively small, the incumbent is willing to bundle the licensed and unlicensed service to gain more profit. However, when $B$ is large, even if there is plenty of unlicensed spectrum, the incumbent may not want to bundle the service. Intuitively, when the incumbent bundles its service, it has to announce a lower price, which may result in a drop in its profit.

	\subsection{Multiple incumbent SPs and no entrant SPs}
	Here we consider the case with $M$ incumbent SPs and no entrants. 
	
	
	To simplify the model, we consider linear case with $M$ symmetric SPs. Suppose the total amount of licensed spectrum is $B_t$, then we have $B_i =  \frac{B_t}{M}$, for $\forall i$. In this case, each SP is equivalent. Hence by Corollary \ref{col:symmetric_equilibrium}, a symmetric equilibrium exists, which we focus on in the following. Further, to simplify our analysis in this part we assume e a common choice of $\alpha$ is agreed upon by the SPs to either maximize their welfare or to maximize social welfare.
	
	We first fix  $M$ and let $W\to\infty$. 
	\begin{theorem}
		\label{thm:optimize_profit}
		In the linear symmetric case, if $W\to\infty$, there exists some $B_{th}$, such that when $B_t\le B_{th}$, the $\alpha$ that maximizes the profit of incumbents is $1-\sqrt{\frac{B_t}{B_{th}}}$ and social welfare decreases with $M$. When $B_t> B_{th}$, the optimal $\alpha$ is $0$ and social welfare increases with $M$.
	\end{theorem}
	
	
	
	Similar to the case with only one incumbent and one entrant SP, Theorem \ref{thm:optimize_profit} shows that when there are multiple incumbent SPs, they may only want to bundle the licensed and unlicensed service when they have a limited amount of licensed spectrum. Additionally, when $B_t\le B_{th}$ the optimal announced price is not related to the total licensed bandwidth $B_t$, because the optimal choice of $\alpha$ counteracts the effect of $B_t$ when $W\to \infty$. When the bandwidth of licensed spectrum is above the threshold, even if the unlicensed resources are unlimited, the incumbent SPs may not be willing to use it. In this case, the announced price  is decreasing in the total licensed bandwidth $B_t$. In fact the threshold $B_{th}$ is  decreasing  in the total number of incumbent SPs in the market. Intuitively, when more incumbent SPs are in the market, competition on the unlicensed band would be more intense. The incumbent SPs are more willing to provide service on licensed spectrum only. 
	
	For the social welfare part of Theorem \ref{thm:optimize_profit}, when the licensed spectrum is quite limited, adding more incumbent SPs into the market harms the social welfare, because more competitors will cause $\alpha$ to decrease. That means less unlicensed resources are utilized and as a result the social welfare decreases. However, when there is plenty of licensed spectrum, all the SPs have already chosen not to use the unlicensed spectrum. Adding more competitors then benefits the social welfare.

	Next we consider the case where the incumbent SPs are  altruistic, which means they would like to achieve a Nash equilibrium that maximizes the social welfare.\footnote{Alternatively, the choice of $\alpha$ could be chosen by  a regulator who seeks to maximize welfare.}   We then have the following result.
	\begin{lemma}
		\label{thm:optimize_sw}
		If $W\to\infty$,  $\alpha = 1$ is socially optimal.
	\end{lemma}
	
	When $\alpha = 1$, all the prices go to $0$ and all customers get served, which is socially optimal. Essentially, the  SPs are pushing all the customers to the unlicensed band. This yields zero profits but achieves the optimal welfare.

	\subsection{Social welfare gap}
	Next, we consider the gap in social welfare between the profit maximizing choice of $\alpha$ and the social optimal $\alpha$.
	
	We consider a similar scenario as in the previous subsection. 
	\begin{theorem}
		\label{thm:sw_gap}
		In the case of infinite unlicensed spectrum, if we also assume an infinite number of incumbent SPs, i.e., $M\to \infty$, the social welfare gap between the profit optimal and welfare optimal settings is given by 
		\begin{equation}
		\label{eqn:sw_gap}
		Gap = \frac{1}{2+\max\{2,B_t\}}.
		\end{equation}
	\end{theorem}

	Theorem \ref{thm:sw_gap} shows that when the total licensed bandwidth $B_t$ is small and $M\to\infty$, the social welfare gap is independent of $B_t$. That is because the SPs are able to choose appropriate $\alpha$ to bundle the licensed and unlicensed service, which actually compensates the effect of the amount of licensed bandwidth. However, when the licensed band $B_t$ is large, the incumbent SPs only use the licensed band and as a result the social welfare gap decreases with $B_t$. Also notice that  when the number of incumbent SPs goes to infinity, the social welfare gap between the optimizing profit and optimizing social welfare is upper bounded by $\frac{1}{4}$, which means the price of anarchy is lower bounded by $\frac{1}{2}$ .

	
	\section{Numerical Results}
	\label{sec:numerical_results}
	In this section we give some numerical examples illustrating our results. We first consider a linear model where $P(x) = 1-x$, $g(x) = x$ with one incumbent SP and one entrant SP. 
	
	First, we examine how the optimal $\alpha$ varies with the additional unlicensed spectrum when the bandwidth $B$ of licensed spectrum is fixed. The results are shown in Fig. \ref{fig:opt_alpha}.  We consider two different cases with $B=1$ and $B=3$, respectively. When bandwidth of licensed spectrum is relatively small, i.e. $B = 1$, we can see that the $\alpha$ that optimizes the profit of incumbent will first increase and then decrease  with $W$. That is because when $B$ and $W$ are both small, the incumbent SP is willing to offload some of the traffic to the unlicensed band without harming its profit on the licensed band. However, when $W$ becomes larger,  the incumbent will have to face more competition with the entrant SP, which may lower the delivered price in the market and consequently harm its profit on the licensed band. As a result the optimal $\alpha$ will be deceasing in $W$. Note that the optimal $\alpha$ never drops to $0$ when $B=1$ but does reach $0$ when $B=3$. That illustrates the results in Theorem \ref{thm:alpha_one_in_one_en}. As we know when $\alpha = 0$, the bundling case reduces to the exclusive use case. This means that when $B$ is small, we are able to find some $\alpha\in(0,1)$ such that the profit of incumbent in the bundling case is higher than that in the exclusive use case. But when $B$ is large and $W$ is in some range, the exclusive use case may provide the best profit for the incumbent. This property coincides with the results in Theorem \ref{thm:exclusiveprofit}. In the case that the incumbent SP aims to optimize the social welfare, the optimal $\alpha$ is increasing with the unlicensed bandwidth $W$. It is not surprising because sending more traffic on the wider band may help lower the congestion cost in the market and lead to better social welfare. Note that when $B=1$, there is an intersection of the solid and dashed curves around $W=1$, which indicates both the incumbent's profit  and social welfare are maximized at the same $\alpha$. Furthermore, when $B$ increases to $3$, the intersection point appears when $W$ is smaller, because when the incumbent SP has enough licensed spectrum, it may not be willing to enter the unlicensed market which will lead to a loss of social welfare.
	\begin{figure}[htbp]
		\centering
		\includegraphics[scale=0.5]{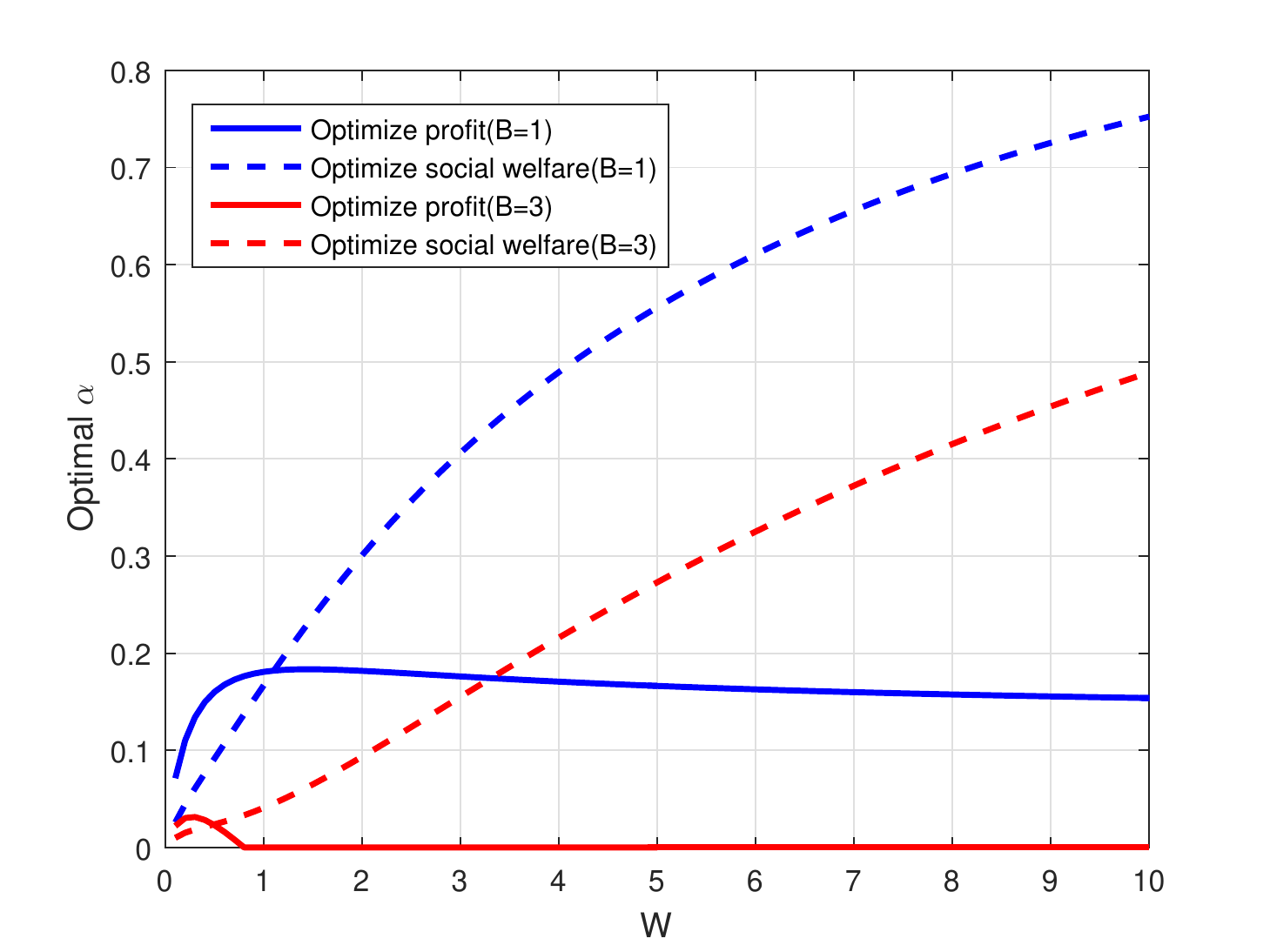}
		\caption{The optimal $\alpha$ chosen by the incumbent SP.}
		\label{fig:opt_alpha}
	\end{figure}
	
	
	Next, we focus on the case with $B=1$ to see how profit and social welfare compares with the unbundled and exclusive use cases when $\alpha$ is chosen to optimize either profit or social welfare. We first look at the profit optimal case.  The results are shown in Fig. \ref{fig:opt_profit}. We can see that when the incumbent's profit is optimized, bundling can obviously achieve the best profit. At the same time, for a wide range of unlicensed bandwidth, the social welfare of bundling is better than both the unbundled  and exclusive access cases. 
	\begin{figure}[htbp]
		\centering
		\subfigure[Profit of incumbent ]{
			\label{fig:opt_profit:a} 
			\includegraphics[width=1.68in]{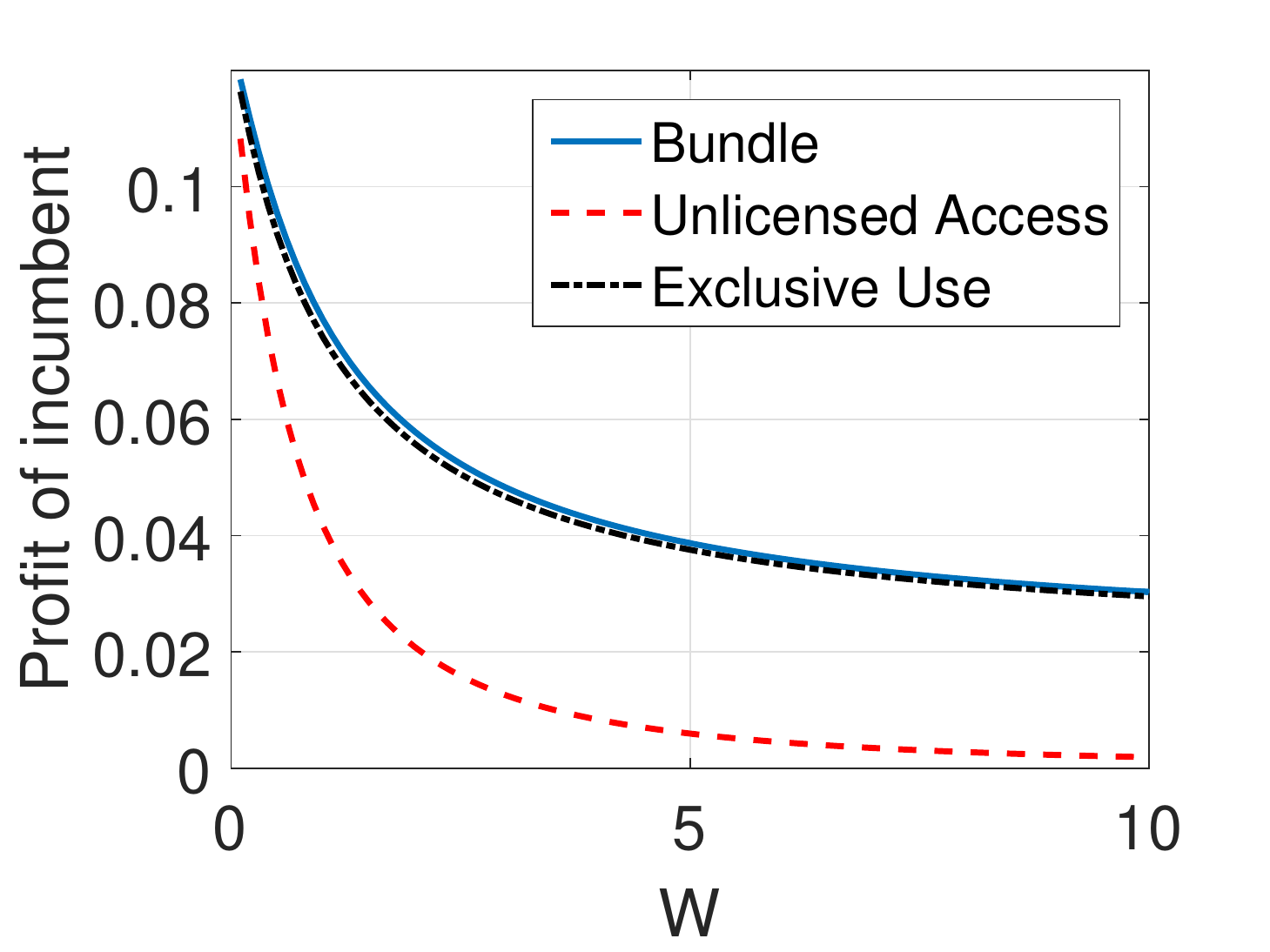}}
		\subfigure[Social welfare]{
			\label{fig:opt_profit:b} 
			\includegraphics[width=1.68in]{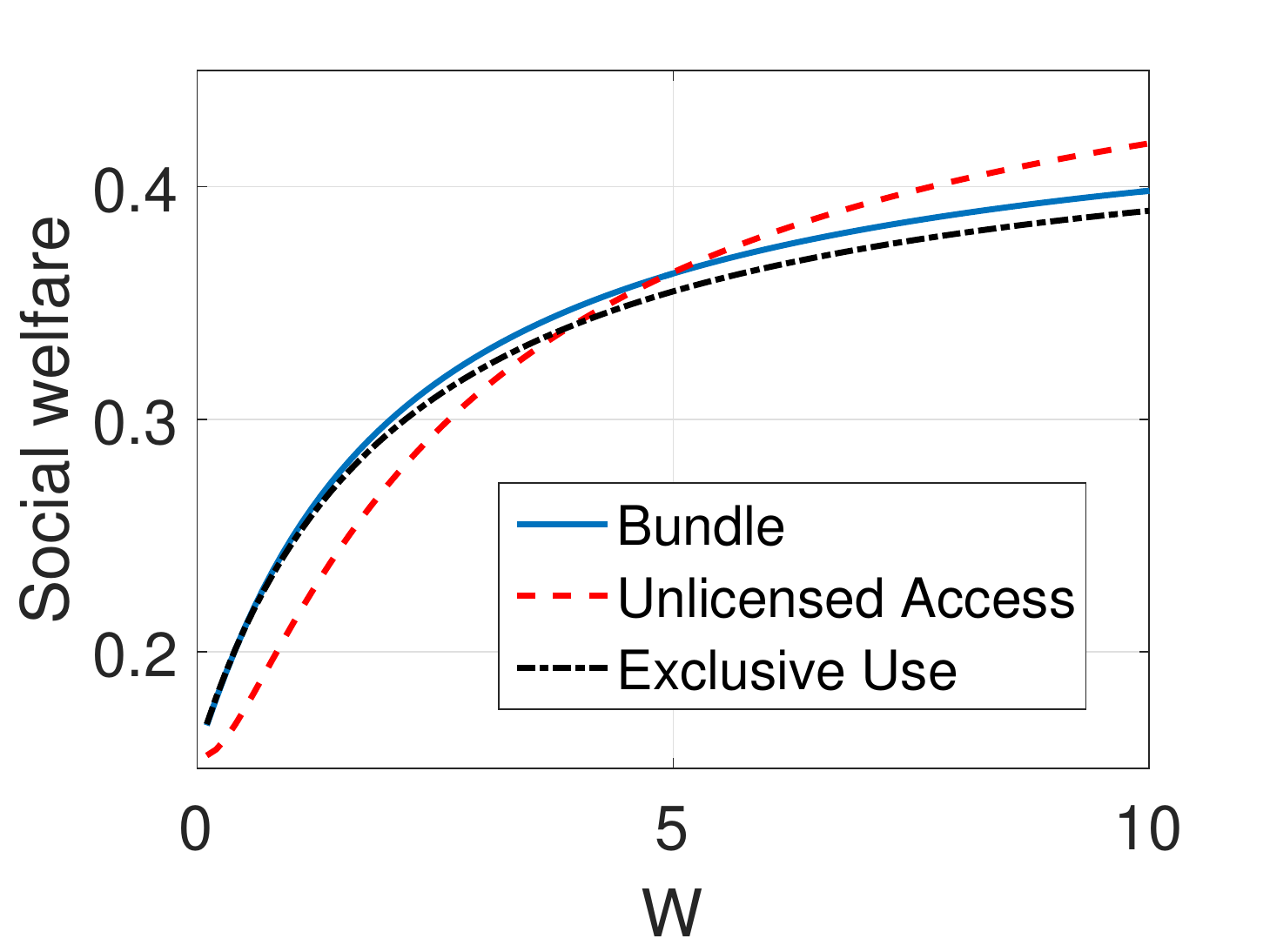}}
		\caption{Impact of choosing  $\alpha$ to maximize the incumbent's profit  when $B=1$.  }
		\label{fig:opt_profit} 
	\end{figure}
	
	\begin{figure}[htbp]
		\centering
		\subfigure[Profit of incumbent ]{
			\label{fig:opt_sw:a} 
			\includegraphics[width=1.68in]{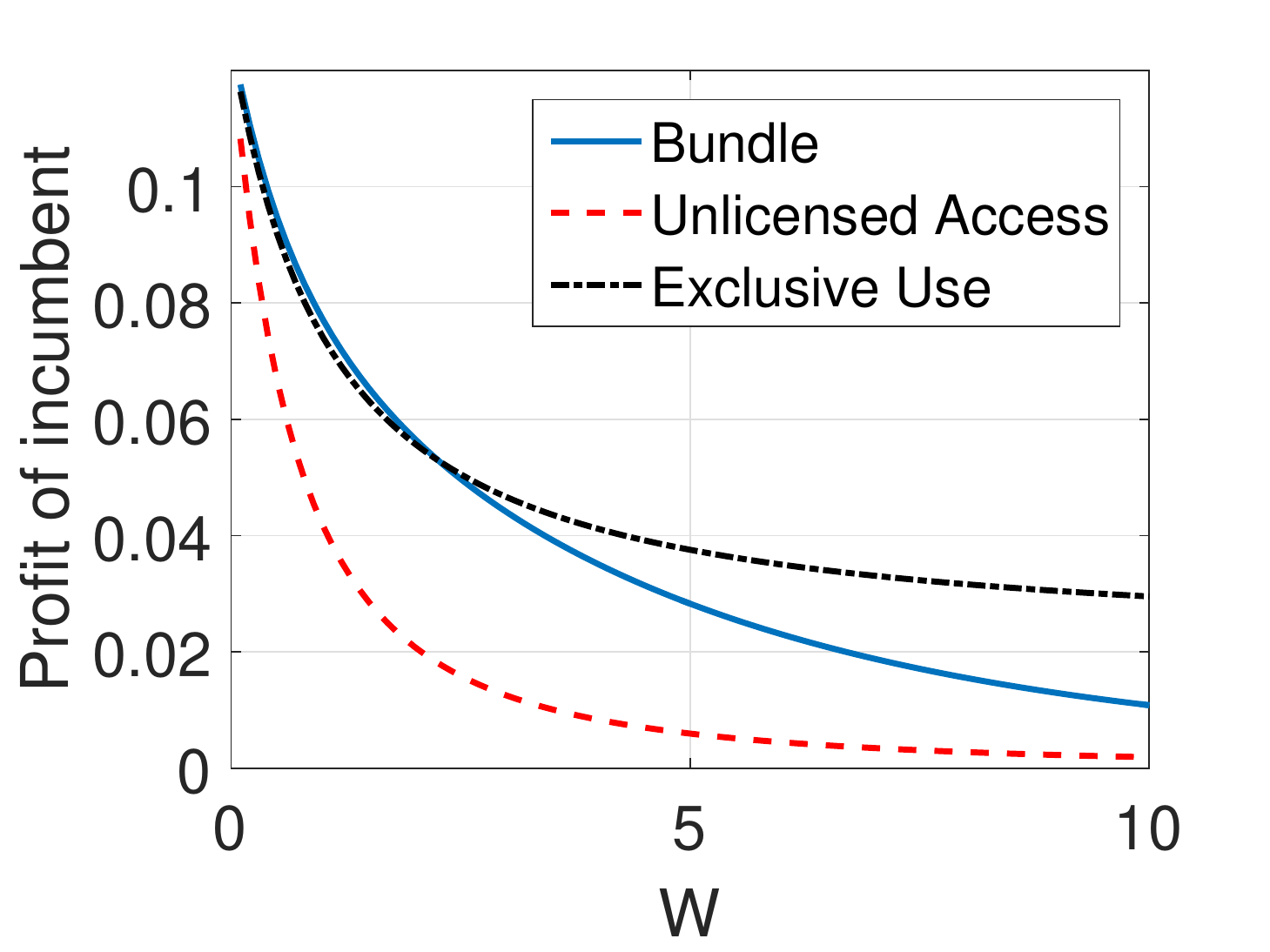}}
		\subfigure[Social welfare]{
			\label{fig:opt_sw:b} 
			\includegraphics[width=1.68in]{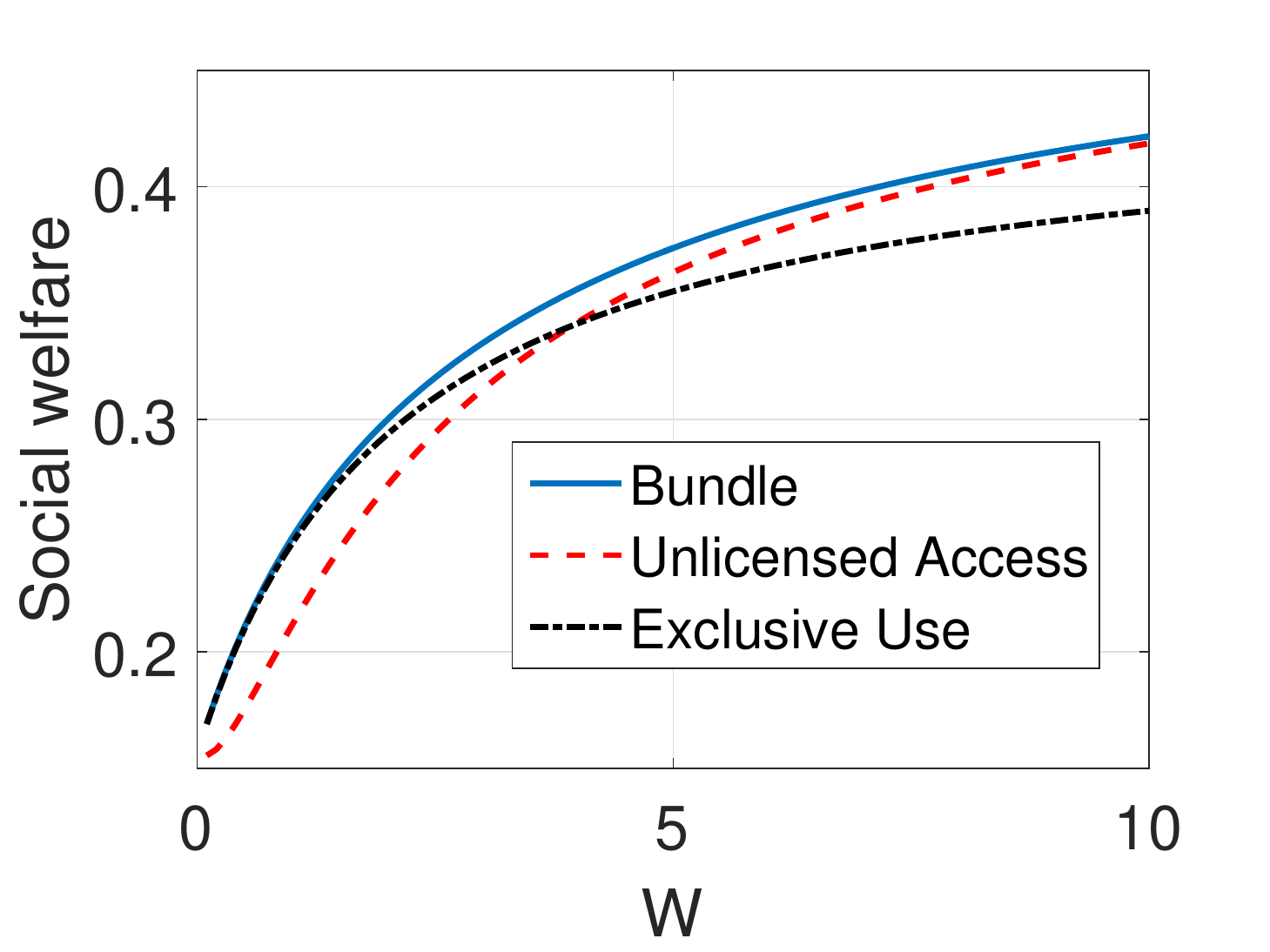}}
		\caption{Impact of choosing  $\alpha$ to maximize the social welfare when $B=1$.  }
		\label{fig:opt_sw} 
	\end{figure}
	
	
	Next we consider the welfare optimal case.  The results are shown in Fig. \ref{fig:opt_sw}. When social welfare is maximized, there is still some range of $W$ such that the incumbent's profit with bundling  is the highest among the three cases. But when $W$ increases, the profit with  bundling is converging to that without bundling. As discussed previously, the bundling case is a kind of generalization of the unbundled  and exclusive use cases.  Further, it can be seen from the figures that it is possible to use the bundling strategy to improve both the profit of incumbent and the social welfare at the same time.

	
	Next we consider the case with $M$ symmetric incumbent SPs and no entrant SP again assuming linear congestion and demand.  We first fix the total bandwidth of licensed spectrum $B_t$ to see how the social welfare gap between the social welfare optimizing and the profit maximizing cases changes with the unlicensed bandwidth $W$.  The results are shown in Fig. \ref{fig:sw_gapB}. In general, the social welfare gap increases with $W$. When $W$ increases, the social welfare maximizing case tends to serve more customers to improve  welfare while the profit maximizing case tends to keep the price at a certain level so that the incumbent's profits are not affected. As a result, the gap is getting larger. This implies the social welfare gap we obtain in Theorem \ref{thm:sw_gap} is actually an upper bound. In Fig. \ref{fig:sw_gapB}(a) where $B_t=1$, the social welfare gap increases with  the number of incumbent SPs, because  when the number of incumbents increases, there is more competition in the market which is good for the social welfare. But for the profit maximization case, the SPs limit the customer mass served on the unlicensed band to make their profits higher. This makes the social welfare worse. However, when the total licensed bandwidth increases to $B_t=3$, Fig. \ref{fig:sw_gapB}(b) shows that increasing $M$ actually makes the social welfare gap smaller.  That is because when there is plenty of licensed spectrum, the SPs are not using the unlicensed spectrum in the profit maximizing case. As a result, adding more unlicensed spectrum to the market benefits the social welfare. 
	\begin{figure}[htbp]
		\centering
		\subfigure[$B_t=1$ ]{
			\label{fig:sw_gapB:a} 
			\includegraphics[width=1.68in]{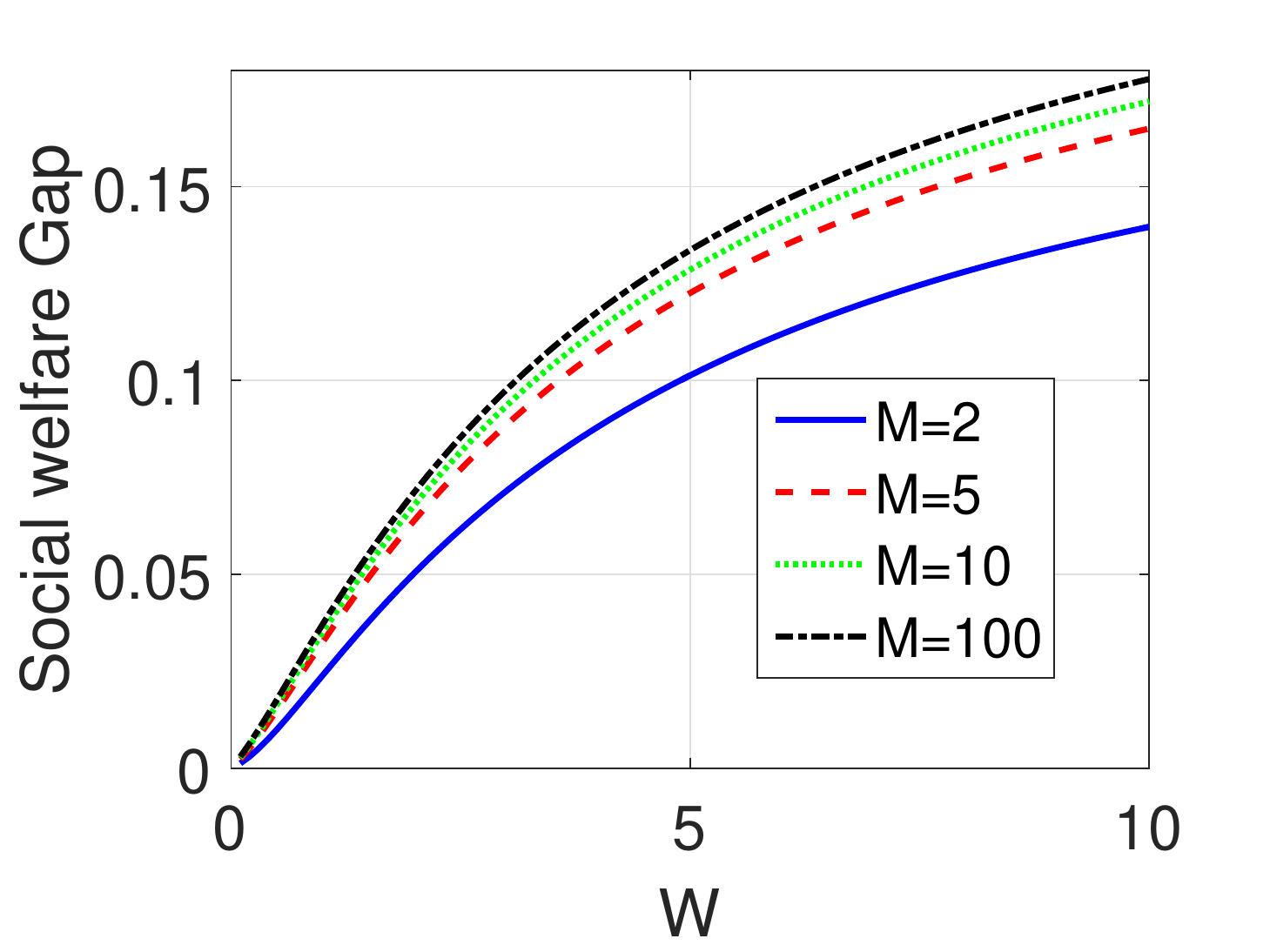}}
		\subfigure[ $B_t=3$]{
			\label{fig:sw_gapB:b} 
			\includegraphics[width=1.68in]{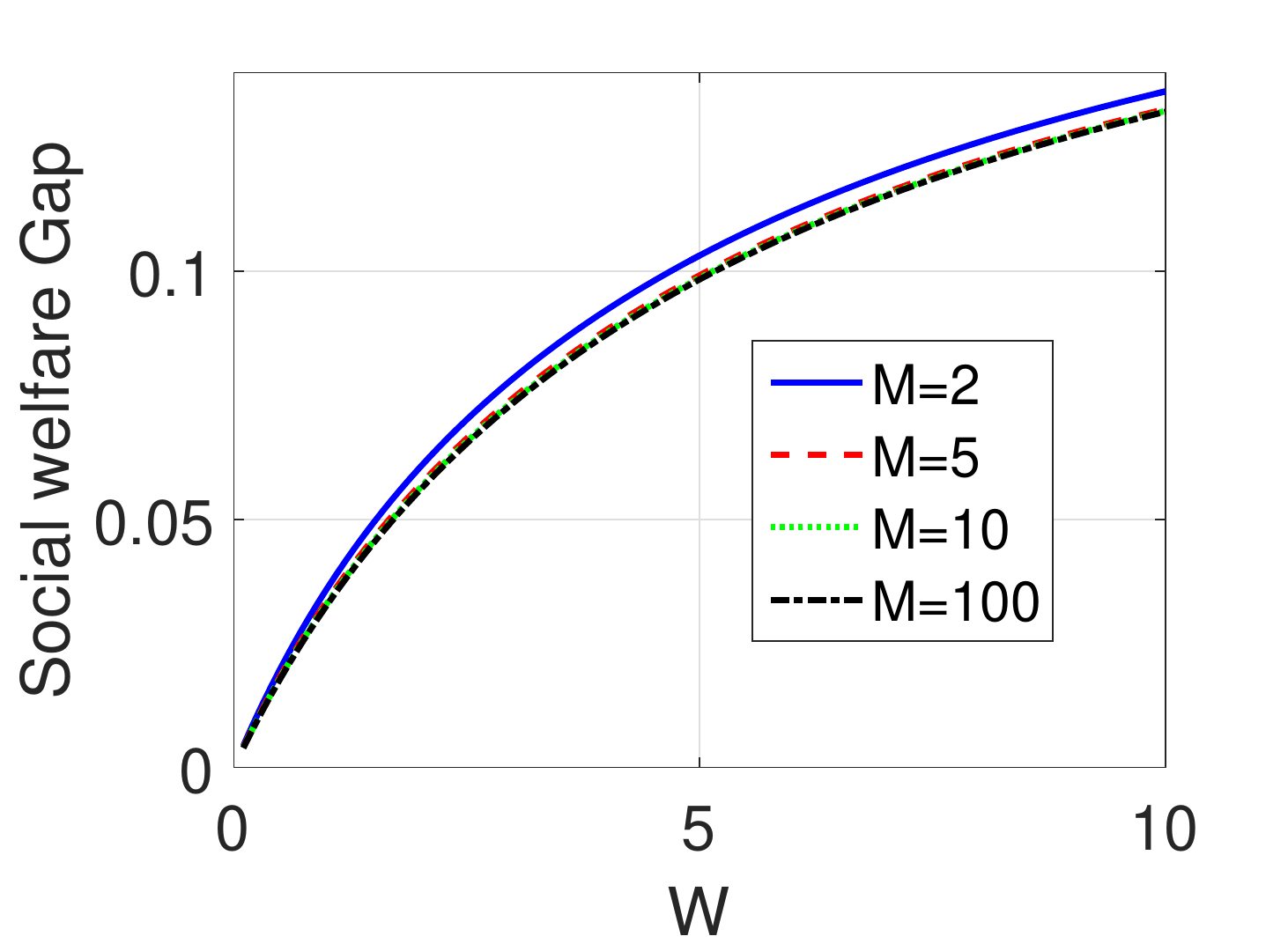}}
		\caption{Social welfare gap when $B_t$ is fixed. }
		\label{fig:sw_gapB} 
	\end{figure}
	
	
	Next we fix the bandwidth of unlicensed spectrum to see how the social welfare gap changes with the total licensed bandwidth $B_t$. Results are shown in Fig. \ref{fig:sw_gapW}. In each of these two figures when either $W=1$ or $W\to\infty$, there is a turning point in each curve, which corresponds to the threshold $B_{th}$ in Theorem \ref{thm:optimize_profit}. When we fix $W = 1$, the social welfare gap first increases then decreases with licensed bandwidth $B_t$ and number of SPs $M$. When the unlicensed bandwidth is increased to infinity, we can see that the social welfare gap remains constant when $B$ is small and decreases afterwards. The constant region appears because the profit maximizing $\alpha$ is chosen to compensate for the change of $B_t$ and the social welfare remains constant in the social welfare maximizing case. This is consistent with the results in Theorem \ref{thm:optimize_profit} and Theorem \ref{thm:sw_gap}.
	\begin{figure}[htbp]
		\centering
		\subfigure[ $W=1$ ]{
			\label{fig:sw_gapW:a} 
			\includegraphics[width=1.68in]{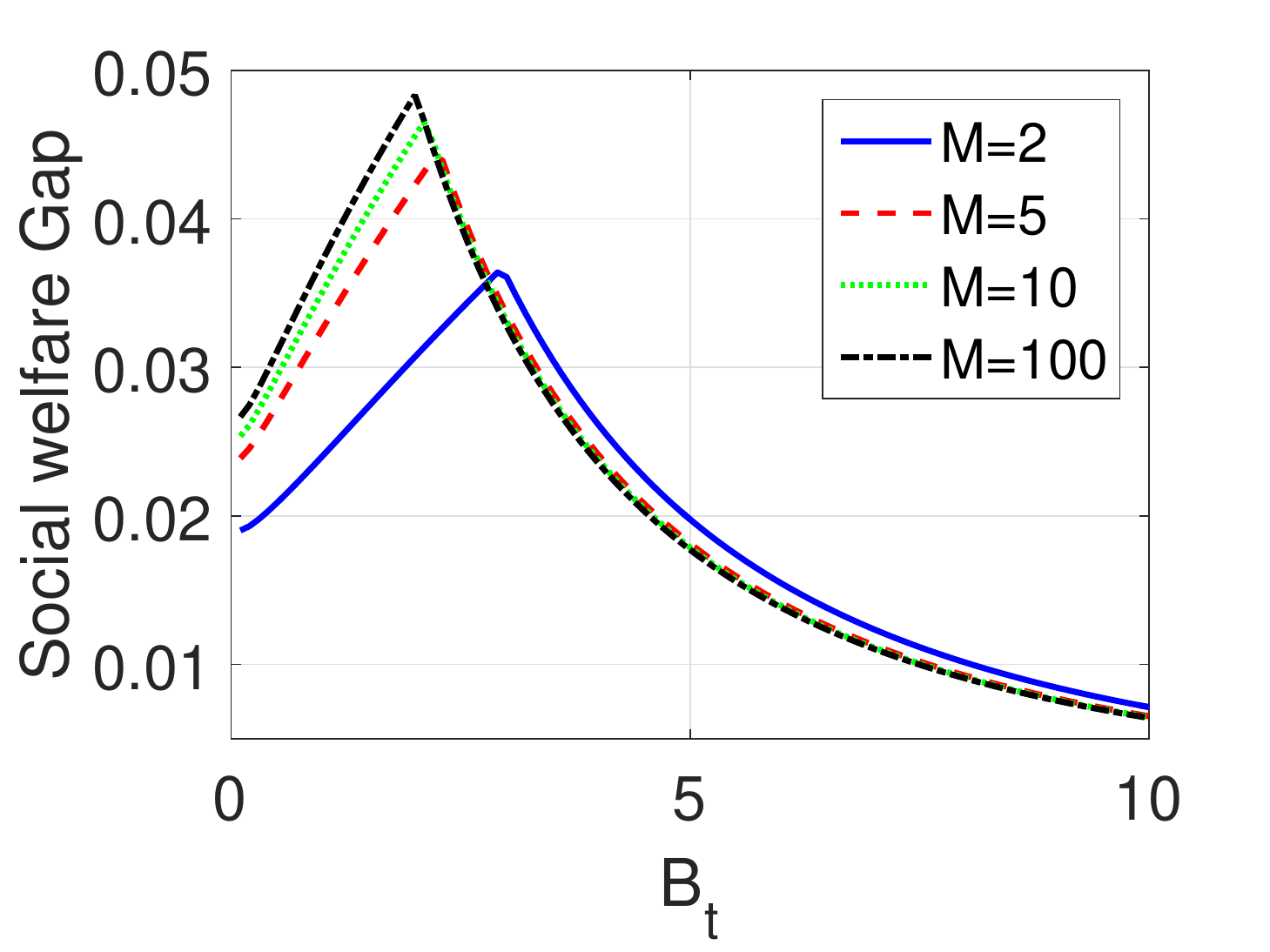}}
		\subfigure[  $W\to\infty$]{
			\label{fig:sw_gapW:b} 
			\includegraphics[width=1.68in]{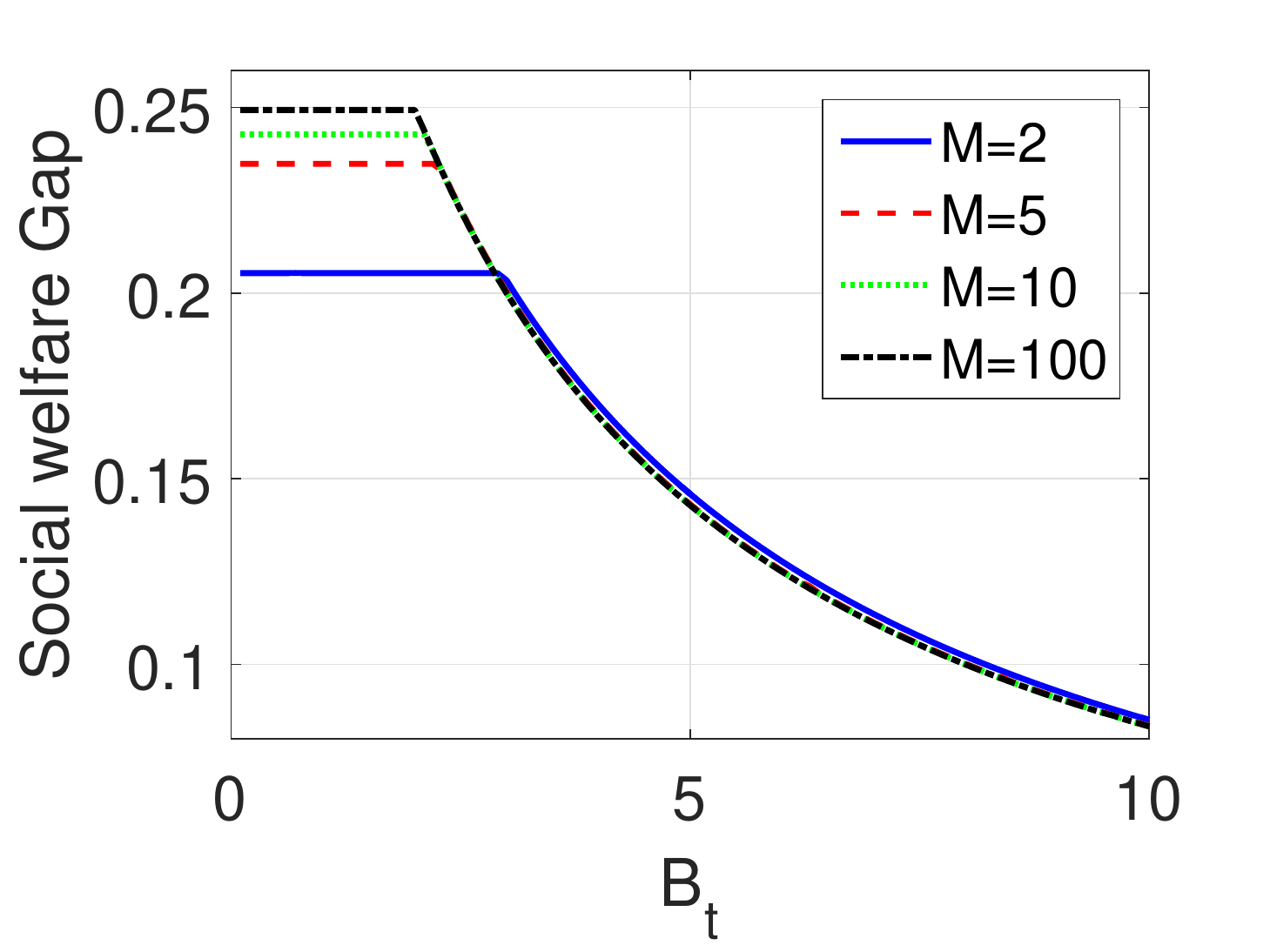}}
		\caption{Social welfare gap when $W$ is fixed.  }
		\label{fig:sw_gapW} 
	\end{figure}

	
	\section{Conclusion}
	\label{sec:conclusion}
	In this paper, we considered the use of bundling in the competition among incumbent and entrant SPs on licensed and unlicensed spectrum. We first analyzed the case where the users' average percentage of time on the unlicensed band $\alpha$ is fixed. We compared the bundling case with the unlicensed access case used in \cite{nguyen2011impact,nguyen2015free} and the exclusive use case with one incumbent and one entrant. We showed that compared to these models, bundling  can improve the profits of the incumbent SP as well as those of the entrant.   We also analyzed the case with multiple incumbents and no entrant. We showed that for a linear class of congestion costs and demands, the game is supermodular and thus the profit of all incumbent SPs can be improved over the unbundled case. We then viewed $\alpha$ as a parameter controlled by the SPs.  We characterized  the social welfare gap between the profit maximizing and the social welfare optimizing cases.
	
	In our work, only homogeneous customers were considered.  We only use one parameter $\alpha$ to characterize the behavior of all users. Exploring the case with heterogeneous customers is one  possible extension. Furthermore, when finding the social welfare gap, we only considered  the symmetric case, where each SP possesses the same amount of licensed spectrum, so that the choice of $\alpha$ reduces to an optimization problem.  It would be also interesting to investigate the asymmetric case where each SP can play an extensive form game on both $\alpha$ and the price.

	%
	%
	
	\bibliographystyle{IEEEbib}
	\bibliography{mybib}
	
   \newpage

	\appendices
	
	\section{Proof of Theorem \ref{thm:mono}}
	\begin{proof}
		For the optimization problem in (\ref{eqn:monoprice}), we can show that $p^l_1 = p^u_1$. It follows that  customers suffer the same amount of congestion on both bands, i.e. $\frac{x^l_1}{B} = \frac{x^u_1}{W}$. Using this,  the optimization problem in (\ref{eqn:monoprice}) can be rewritten as 
		\begin{eqnarray}
		\label{eqn:monoprice1}
		\max_{x_t}&&p_tx_t\\
		{\rm s.t.} && p_t+g\left(\frac{x_t}{B+W}\right)= P(x_t),\nonumber\\
		&&  x_t\ge0. \nonumber
		\end{eqnarray}
		We focus on the congestion part. The congestion for the bundling problem is $c_1(x) =(1-\alpha)g\left(\frac{(1-\alpha)x}{B}\right)+\alpha g\left(\frac{\alpha x}{W}\right) $ and the congestion for the unbundled case is $c_2(x )=g\left(\frac{x_t}{B+W}\right)$. The second order derivative of $c_1(x)$ with respect to $\alpha$ is
		\begin{small}
			\begin{eqnarray}
			\frac{\partial^2 c_1(x)}{\partial^2\alpha} &=&\frac{x}{B}g'\left(\frac{(1-\alpha)x}{B}\right)+\frac{x}{W}g'\left(\frac{\alpha x}{W}\right)\nonumber\\
			&&+\frac{x}{B}\left[g'\left(\frac{(1-\alpha)x}{B}\right)+(1-\alpha) \frac{x}{B}g''\left(\frac{(1-\alpha)x}{B}\right)\right]\nonumber \\
			&&+\frac{x}{W}\left[g'\left(\frac{\alpha x}{W}\right)+\alpha \frac{x}{W}g''\left(\frac{\alpha x}{W}\right)\right]. \nonumber
			\end{eqnarray}
		\end{small}
		Since $g(x)$ is convex increasing, it follows that  $\frac{\partial^2 c_1(x)}{\partial^2\alpha} \ge 0$. Thus for any given $x$, $c_1(x)$ is a convex function of $\alpha$. Note that $\alpha = \frac{W}{B+W}$ is a solution to the equality $\frac{\partial c_1(x)}{\partial\alpha} =0$ and when $\alpha = \frac{W}{B+W}$, $c_1(x) = c_2(x)$ That means for any given $x$, $c_1(x) \ge c_2(x)$.  
		
		As illustrated in Fig. \ref{fig:mono1}, the two optimization problems in (\ref{eqn:monoprice}) and (\ref{eqn:monoprice1}) are each seeking a rectangle with maximum area in the constrained region between the inverse demand and congestion curve. Since $c_1(x) \ge c_2(x)$, the profit of bundling cannot exceed that of unlicensed access.
		\begin{figure}[htbp]
			\centering
			\includegraphics[scale=0.25]{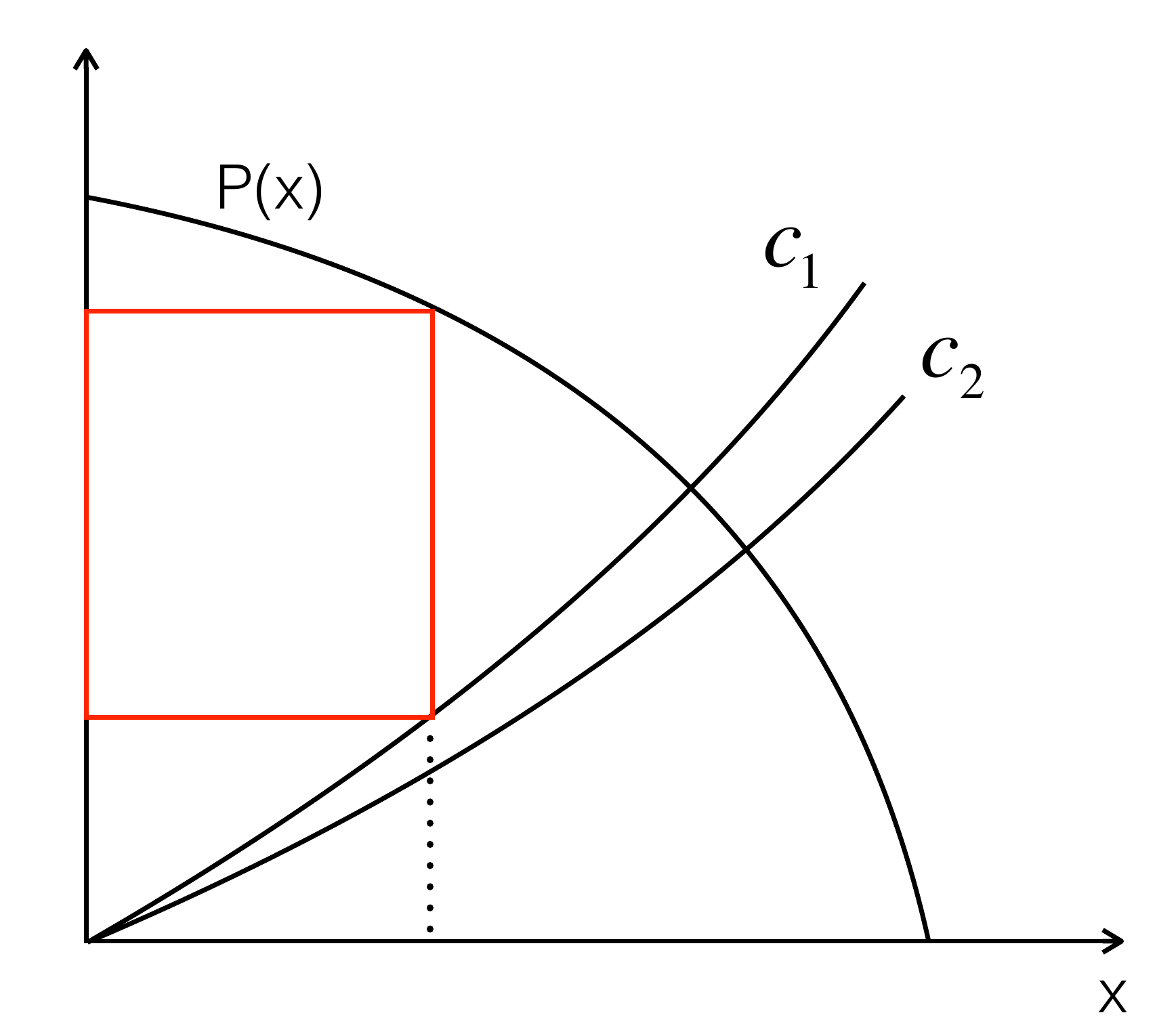}
			\caption{Illustration of the proof of Theorem \ref{thm:mono}}
			\label{fig:mono1}
		\end{figure}
	\end{proof}
	
	\section{Proof of Theorem \ref{thm:incumbent_entrant1}}
	\begin{proof}
		Consider the unbundled case, we have already know that the price on the unlicensed band is $0$ from Lemma \ref{lemma:price}. Let $p^{*}_1,x^{*}_1$ be the equilibrium price and customer mass on the licensed band, respectively for the incumbent SP in the unbundled case. And let $w^*_t$ be equilibrium total customer mass on the unlicensed band in the unbundled case. 
		
		Let $\alpha_0 = \frac{w^*_t}{w^*_t+x^*_1}$. In the unbundled case, for any given $B$ and $W$, we have $x^*_1>0$ and $w^*_t >0$, thus $\alpha\in(0,1)$. It can be verified that $\hat{p}_1 = (1-\alpha_0) p^*_1$ is the solution to the following optimization problem given $\hat{p}_2 = 0$, i.e. the best response of the incumbent SP given entrant announce price $0$.
		\begin{eqnarray}
		\label{eqn:incumentrant1}
		\max_{p_1}&&p_1x_1\\
		{\rm s.t.} && p_1+(1-\alpha_0)g\left(\frac{(1-\alpha_0)x_1}{B}\right)+\alpha_0 g\left(\frac{\alpha_0x_1 + x_2}{W}\right) \nonumber\\
		\label{constraint1}
		&&= P(x_1+x_2),\\
		\label{constraint2}
		&& p_2+ g\left(\frac{\alpha_0x_1 + x_2}{W}\right) = P(x_1+x_2),\\
		&&  p_1\ge0, p_2\ge0. \nonumber
		\end{eqnarray}
		Next we show that $\hat{p}_2 = 0$ is also the best response of the entrant SP given $p_1 = (1-\alpha_0) p^*_1$, which implies $\hat{p}_1 = (1-\alpha_0) p^*_1$ and $\hat{p}_2 = 0$ is a Nash equilibrium of the game. 
		
		The resulting customer mass served by the SPs given the announced price  $\hat{p}_1 = (1-\alpha_0) p^*_1$ and $\hat{p}_2 = 0$ can be written as  $\hat{x}_1 = x^*_1+w^*_t$ and  $\hat{x}_2=0$ . They must satisfy the constraints in (\ref{constraint1}) and (\ref{constraint2}).  Given $p_1 = (1-\alpha_0) p^*_1$, suppose the best response of the entrant SP is some $\tilde{p}_2>0$, and let $\tilde{x}_1$, $\tilde{x}_2$ be the corresponding customer mass. Then we must have $\tilde{x}_1\le \hat{x}_1$, because otherwise constraint (\ref{constraint2}) will be violated. We can also conclude that $\alpha_0\tilde{x}_1 +\tilde{x}_2\ge \alpha_0\hat{x}_1+\hat{x}_2$, because otherwise constrain (\ref{constraint1}) will be violated. Note that the left hand side of constraints (\ref{constraint1}) and (\ref{constraint2}) must also be equal. That gives us the following equality constraint
		\begin{small}
			\begin{eqnarray}
			\label{eqn:constraint}
			p_1 +(1-\alpha_0)g\left(\frac{(1-\alpha_0)x_1}{B}\right)=p_2+ (1-\alpha_0)g\left(\frac{\alpha_0x_1 + x_2}{W}\right). 
			\end{eqnarray}
		\end{small}
		It is obvious when $\tilde{p}_1 = \hat{p}_1$, $\tilde{p}_2>\hat{p}_2$, $\tilde{x}_1\le \hat{x}_1$ and $\alpha_0\tilde{x}_1 +\tilde{x}_2\ge \alpha_0\hat{x}_1+\hat{x}_2$, constraint (\ref{eqn:constraint}) cannot hold at $\tilde{p}_2>0$. So $p_2 = 0$ is the best response to $p_1 = (1-\alpha_0) p^*_1$  which makes it the equilibrium of the game.  It then can be verified that the profit of the SPs and the welfare are identical to the unbundled case.
	\end{proof}
	
	\section{Proof of Theorem \ref{thm:incumbent_entrant2}}
	\begin{proof}
		First consider the case when $\alpha<\alpha_0$.  Let $p^{*}_1,x^{*}_1$ be the equilibrium price and customer mass served by the incumbent SP respectively  in the bundling case. And let $x^*_2$ be equilibrium customer mass served by the entrant SP. Based on the definition of Nash equilibrium, $p^*_1$ is the best response of the incumbent SP given $p_2 = p^*_2$. 
		
		Let $\tilde{p}_1$ be the best response of the incumbent SP given $p_2 = 0$ and $\tilde{x}_1$ be the corresponding customer mass served by the incumbent. Then $\tilde{p}_1$ is the solution to the optimization problem in (\ref{eqn:incumentrant1}) given  $p_2 = 0$. It can be shown that $\tilde{p}_1 \tilde{x}_1$ is the same as the profit of the incumbent SP in the unbundled case. And given the condition $\alpha<\alpha_0$, we have $\tilde{x}_2 > 0$.
		
		What we need to show is that $\tilde{p}_1 \tilde{x}_1\le p^{*}_1 x^{*}_1$.	We prove this statement by contradiction. Suppose we have $\tilde{p}_1 \tilde{x}_1> p^{*}_1 x^{*}_1$. We consider two cases when the entrant SP change his strategy from $p_2 = 0$ to $p_2 = p^*_2$. First case is $p^*_2 +g\left(\frac{\alpha \tilde{x}_1}{W}\right)\ge P(\tilde{x}_1)$, which means the entrant SP can attract no customers by setting his price to $p^*_2$. That means $w_2$ changes from $\tilde{w}_2$ to $0$. Under this condition, the incumbent SP can  still increase its announced price to some $\tilde{p}_1'$ to satisfy the constraints in (\ref{constraint1}) without changing  the customer mass served $x_1$. Consequently we have $\tilde{p}_1'\tilde{x}>\tilde{p}_1\tilde{x}>p^{*}_1 x^{*}_1$, which violates that $p^*_1$ is the best response given $p_2 = p^*_2$.
		
		Another case is when $p^*_2 +g\left(\frac{\alpha \tilde{x}_1}{W}\right)< P(\tilde{x}_1)$, which means it is possible for the entrant SP to serve positive amount of customers. Let $\tilde{x}_2'$ be the customer mass that satisfies  $p^*_2 +g\left(\frac{\alpha \tilde{x}_1+\tilde{x}_2'}{W}\right)= P(\tilde{x}_1+\tilde{x}_2')$. Since $p^*_2>0$, we must have $\tilde{x}_2'<\tilde{x}_2$. That implies the incumbent SP can still increase its announced price to $\tilde{p}_1'>\tilde{p}_1$ without changing $x_1$. Therefore we have $\tilde{p}_1'\tilde{x}>\tilde{p}_1\tilde{x}>p^{*}_1 x^{*}_1$, which also contradicts the optimality of $p^*_1$ given $p_2 = p^*_2$. In conclusion we always have $\tilde{p}_1 \tilde{x}_1\le p^{*}_1 x^{*}_1$. That means when $\alpha<\alpha_0$, the incumbent SP can get more profit by bundling than competing without bundling. Similar argument can be used to show that the entrant SP can gain positive profit when $\alpha<\alpha_0$.
		
		Then we consider the case when $\alpha>\alpha_0$. We have already shown that when $\alpha = \alpha_0$, the entrant SP announces $0$ price and serves no customers in equilibrium in the proof of Theorem \ref{thm:incumbent_entrant1}. To show that the entrant get $0$ profit when $\alpha>\alpha_0$, we use $p_1^{(0)}$ and $x_1^{(0)}$ to denote the equilibrium price and customer mass of the incumbent SP when $\alpha = \alpha_0$. Assume when $\alpha>\alpha_0$, we end up with equilibrium price and customer mass of entrant SP $\tilde{p}_2>0, \tilde{x}_2>0$.  If $\tilde{x}_1+\tilde{x}_2 \ge x_1^{(0)} $,  the delivered price get lowered, thus we have 
		\begin{equation}
		\label{eqn:inequality1}
		\alpha\tilde{x}_1+\tilde{x}_2 < \alpha_0 x_1^{(0)}.
		\end{equation}
		 Then we must have $\tilde{x}_1<x_1^{(0)}$, which implies 
		 \begin{equation}
		 \label{eqn:inequality2}
		(1-\alpha)\tilde{x}_1<(1-\alpha_0)x_1^{(0)}.
		 \end{equation} 
		Adding (\ref{eqn:inequality1}) and (\ref{eqn:inequality2}), we get contradiction to our assumption that $\tilde{x}_1+\tilde{x}_2 \ge x_1^{(0)} $. 	In another case, if $\tilde{x}_1+\tilde{x}_2 < x_1^{(0)} $, we can use similar argument to find contradiction. So we can conclude that when $\alpha>\alpha_0$, the entrant SP can not serve customers with positive price, which gives the entrant SP $0$ profit. That means the game can be reduced to the following optimization problem. 
		\begin{eqnarray}
		\label{eqn:equivalentopt}
		\max_{p_1}&&{p_1}x^l\\
		{\rm s.t.} && p_1+g\left(\frac{x^l}{B}\right)= P(x^l+x^u),\nonumber\\
		&& g\left(\frac{ x^u}{W}\right) = P(x^l+x^u),\nonumber\\
		\label{additionalconstrain}
		&& \frac{x^l}{x^u} = \frac{1-\alpha}{\alpha},\\
		&&  p_1\ge0. \nonumber
		\end{eqnarray}
		Comparing the optimization problem in (\ref{eqn:equivalentopt}) with that in (\ref{eqn:priceoptimization}), the only difference is the additional constraint (\ref{additionalconstrain}). That means the optimization in the bundling case will end up with the profit no greater than that in the unlicensed case. 
	\end{proof}

	\section{Proof of Theorem \ref{thm:exclusiveprofit}}
	\begin{proof}
		The exclusive use case can be formulated as following in the linear case.
		\begin{eqnarray}
		\label{eqn:exclusive}
		\max_{p_i}&&p_i x_i\\
		{\rm s.t.} && p_1+\frac{x_1}{B}= 1-(x_1+x_2),\nonumber\\
		&&  p_2+ \frac{ x_2}{W} = 1-(x_1+x_2),\nonumber\\
		&&  p_i\ge0, i=1,2. \nonumber
		\end{eqnarray}
		In the linear case, we can find the Nash equilibrium explicitly by finding each other's best response function. The equilibrium price and customer mass can be written as
		\begin{eqnarray}
		&&p^*_1 = \frac{2+2B+W}{4+4B+4W+3BW},\;x^*_1 = \frac{B(1+W)}{1+B+W}p^*_1,\nonumber\\
		&&p^*_2 = \frac{2+B+2W}{4+4B+4W+3BW},\;x^*_2 = \frac{W(1+B)}{1+B+W}p^*_2.\nonumber
		\end{eqnarray}
		It is easy to see that when $\alpha = 0$, bundling is equivalent to the exclusive use case. As a result to prove the theorem, we only need to show that $\lim\limits_{\alpha\to0}\frac{\partial p_1x_1}{\partial \alpha}>0$ for some $B$ and $W$. 
		
		Taking derivative directly we have
		\begin{footnotesize}
			\begin{equation}
			\lim\limits_{\alpha\to0}\frac{\partial p_1}{\partial \alpha}=\frac{-\left[4(1+W)+B^2(8+9W) +B(12+17W+6W^2)\right]}{(4+4B+4W+3BW)^2},\nonumber
			\end{equation}
		\end{footnotesize}
		\begin{footnotesize}
			\begin{equation}
			\lim\limits_{\alpha\to0}\frac{\partial x_1}{\partial \alpha}= \frac{B(1+W)\left[B^2(8+3W)+B(20+19W)+4(3+5W+2W^2)\right]}{(1+B+W)(4+4B+4W+3BW)^2} .\nonumber
			\end{equation}
		\end{footnotesize}
		And we can then find the derivative of profit $	\lim\limits_{\alpha\to0}\frac{\partial p_1x_1}{\partial \alpha}= \frac{\partial p_1}{\partial \alpha} x^*_1+\frac{\partial x_1}{\partial \alpha}p^*_1$.
		\begin{footnotesize}
			\begin{equation}
			\label{eqn:profit_dif}
			\lim\limits_{\alpha\to0}\frac{\partial p_1x_1}{\partial \alpha}=\frac{B(1+W)(2+B+W)\left[2(1+B+W)(4+4W-3BW)\right]}{(1+B+W)(4+4B+4W+3BW)^3}.
			\end{equation}
		\end{footnotesize}
		It can be seen that in (\ref{eqn:profit_dif}), everything is positive if the condition $4+4W-3BW>0$ holds. That gives us the exact condition $B<\frac{4(1+W)}{3W}$.
	\end{proof}
	
	\section{Proof of Theorem \ref{thm:exclusiveSW}}
	\begin{proof}
		First consider the customer welfare. Based on the definition, customer welfare is an non decreasing function in the total customer mass served. Since when $\alpha = 0$,  bundling is equivalent to the exclusive use case, we need to show $\lim\limits_{\alpha\to0}\frac{\partial( x_1+x_2)}{\partial \alpha}>0$. We have 
		\begin{footnotesize}
			\begin{eqnarray}
			\label{eqn:cs_dif}
			\lim\limits_{\alpha\to0}\frac{\partial( x_1+x_2)}{\partial \alpha}= \frac{2B(2+B+2W)(2+3W)}{(4+4B+4W+3BW)^2}>0\nonumber
			\end{eqnarray}
		\end{footnotesize}
		
		Similar to the proof of customer surplus, we only need to show that $\lim\limits_{\alpha\to0}\frac{\partial SW}{\partial \alpha}>0$ for the social welfare part. We have 
		\begin{footnotesize}
			\begin{eqnarray}
			\label{eqn:sw_dif}
			\lim\limits_{\alpha\to0}\frac{\partial SW}{\partial \alpha}&=& \frac{\partial p_1x_1}{\partial \alpha}+ \frac{\partial p_2x_2}{\partial \alpha}+(x^*_1+x^*_2)\left(\frac{\partial x_1}{\partial \alpha}+ \frac{\partial x_2}{\partial \alpha}\right) \nonumber\\
			&=& \frac{B[-3W(1+W)B^2+4W^2+7W+4+\frac{4}{3}(1+W)]}{(1+B+W)(4+4B+4W+3BW)^3}\nonumber\\
			&\ge& \frac{B[4(1+W)-3W][(1+W)B+\frac{1}{3}]}{(1+B+W)(4+4B+4W+3BW)^3}>0.\nonumber
			\end{eqnarray}
		\end{footnotesize}
		As a result, there always exist $\alpha>0$ such that both of the customer welfare and social welfare of bundling is better.
	\end{proof}
	
	\section{Proof of  Theorem \ref{thm:equilibrium_existence}}
	\begin{proof}
		We use Debreu's theorem \cite{debreu1952social} to prove the existence of Nash equilibrium. First, for each of the SPs, the strategy space is $p_i\in[0,P(0)]$, where $P(x)$ is the inverse demand function. It is compact and convex. 
		
		Next we show that $u_i(p_i,p_{-i})$ is continuous in $(p_i,p_{-i})$. Given the strategy vector $\bf{p}$ of the SPs, the customer mass vector $\bf{x}$ forms a Wardrop equilibrium if and only if $\bf{x}$ is the solution to the following optimization problem \cite{menache2011network}.
		\begin{small}
			\begin{align*}
			\label{eqn:wardrop_optimization}
			&\min_{\bf{x}} \int_{0}^{\sum\limits_{i}x_i} \alpha g\left(\frac{\alpha z}{W}\right) dz +\sum\limits_{j}   \int_{0}^{x_j} (1-\alpha)g\left(\frac{(1-\alpha) z}{B_j}\right) +p_j dz \nonumber\\
			&{\rm s.t.}  \sum\limits_j x_j = Q,\;{ \bf{x}}\ge0. \nonumber
			\end{align*}
		\end{small}
		Because the objective function is strictly increasing, given any strategy $\bf{p}$, the Wardrop equilibrium customer mass $\bf{x}$ is unique\cite{menache2011network}. By maximum theorem, the resulting customer mass $\bf{x}$ is continuous in $\bf{p}$. That implies that the profit of each SP is continuous in $\bf{p}$.
		
		At last we show that the profit  $u_i(p_i,p_{-i})$ is a concave function in $p_i$. We know that $\frac{\partial^2 u_i(p_i,p_{-i})}{\partial^2 p_i} =  \frac{\partial^2 u_i(p_i,p_{-i})}{\partial^2 x_i}  \frac{\partial^2x_i}{\partial^2 p_i}$. When we fix  $p_{-i}$, it is easy to show that $\frac{\partial^2 u_i(p_i,p_{-i})}{\partial^2 x_i} \le 0$ by substituting the constraint  (\ref{eqn:constraint_n_incumbent}) into the objective function. The only thing we need to show is that $x_i$ is a convex function in $p_i$. We can directly find that $\frac{\partial^2p_i}{\partial^2 x_i}\le 0$ from the constraints in (\ref{eqn:constraint_n_incumbent}). Next we show that $p_i$ is strictly decreasing on $x_i$, which can guarantee that $x_i$ is convex in $p_i$
		
		Fix $p_{-i}$ and find the Taylor expansion of the constraints around the Wardrop equilibrium point $\bf x^*$ for each incumbent SP 
		\begin{small}
			\begin{align*}
			&\Delta p_i +\frac{(1-\alpha)^2}{B_i}g'\left(\frac{(1-\alpha)x^*_i}{B_i}\right)\Delta x_i\nonumber\\
			&+\frac{\alpha^2}{W} g'\left(\frac{\alpha \sum\limits_kx^*_k}{W}\right)\sum\limits_k\Delta x_k =P'\left(\sum\limits_kx^*_k\right)\sum\limits_k\Delta x_k, \nonumber\\
			&\frac{(1-\alpha)^2}{B_j}g'\left(\frac{(1-\alpha)x^*_j}{B_j}\right)\Delta x_j\nonumber\\
			&+\frac{\alpha^2}{W} g'\left(\frac{\alpha \sum\limits_kx^*_k}{W}\right)\sum\limits_k\Delta x_k =P'\left(\sum\limits_kx^*_k\right)\sum\limits_k\Delta x_k, \forall j\ne i. \nonumber
			\end{align*}
		\end{small}
		Given these linear equations we can find
		\begin{equation}
		\label{eqn:delta_p}
		\Delta p_i = -\frac{\prod_{j=1}^{M} a_j +b\left[\sum\limits_{j\ne i} \frac{\prod_{k=1}^{M} a_k}{a_j} \right]+b\prod_{j\ne i} a_j}{\prod_{j\ne i} a_j +b\left[\sum\limits_{j\ne i} \frac{\prod_{k\ne i} a_k}{a_j} \right] } \Delta x_i,
		\end{equation}
		where $a_j = \frac{(1-\alpha)^2}{B_j}g'\left(\frac{(1-\alpha)x^*_j}{B_j}\right)$ and $b =\frac{\alpha^2}{W} g'\left(\frac{\alpha \sum\limits_kx^*_k}{W}\right) -P'\left(\sum\limits_kx^*_k\right)$. It can be seen from (\ref{eqn:delta_p}) that $p_i$ is strictly decreasing on $x_i$, which implies  $x_i$ is convex in $p_i$. And as a result $u_i(p_i,p_{-i})$ is a concave function in $p_i$. Then the existence of Nash equilibrium follows Debreu's results.
		
		Next we show when $B_i > B_j$, the profit of $i$ is higher than that of $j$.  Consider the strategy 'announce the same price as SP $j$' for SP $i$. Since $B_i>B_j$ and they are sharing the unlicensed band $W$, SP $i$ can then serve more customers than $j$ and as a result gain more profit than $j$. Because this strategy may not be the best response of SP $i$, that means his profit can be even higher. So we come to the conclusion that SP $i$ gets more profit than $j$.
	\end{proof}
	
	\section{Proof of Theorem \ref{thm:supermodular}}
	\begin{proof}
		Without loss of generality, we assume the inverse demand function and congestion function are in the form of $P(x) = A-k_1x$ and $g(x) = k_2 x$ respectively. The $M$ constraints can be written as a matrix form 
		\begin{equation}
		\label{eqn:constraint_matrix}
		{\bf p}	+ {\bf Mx} = A {\mathbbm{1}}  ,
		\end{equation}
		where $\bf p$ is the announced price vector,  $\bf M$ is the matrix with ${k_1+\frac{k_2\alpha^2}{W}+\frac{k_2(1-\alpha)^2}{B_1}}$ on the diagonal and ${{k_1+\frac{k_2\alpha^2}{W}}} $ elsewhere, $\bf x$ is the vector of customer mass and $\mathbbm{1}$ is an all-ones vector. If we find the customer mass vector $\bf x$ according to the price vector ${\bf p}$, the solution is ${\bf x} = {\bf M}^{-1} (A{\mathbbm{1}} -{\bf p})$. In an equilibrium, the announced price should be each other's best response. And to show it is a supermodular game, we need to show that the utility function of each SP has increasing differences. Under the linear case, it is sufficient to show ${\bf M}^{-1}$ has positive entries on the diagonal and negative entries elsewhere. The inverse matrix ${\bf M}^{-1} = \frac{{\rm adj}({\bf M})}{\det({\bf M})}$, where ${\rm adj}({\bf M})$ is the adjugate of matrix ${\bf M}$ and $\det ({\bf M})$ is the determinant of  ${\bf M}$.
		
		First we show that $\det ({\bf M})>0$. We can rewrite matrix ${\bf M}$ as
		\begin{equation}
		\label{eqn:matrix_M}
		{\bf M} = {\rm diag}\left(\frac{k_2(1-\alpha)^2}{B_i}\right) + {\mathbbm{1}}{\mathbbm{1}}^T\left(k_1+\frac{k_2\alpha^2}{W}\right).
		\end{equation}
		The determinant is as a result 
		\begin{equation}
		\label{eqn:determinant}
		\det({\bf M})=\left[1+\left(k_1+\frac{k_2\alpha^2}{W}\right)\frac{\sum\limits_{i=1}^{M}B_i}{k_2(1-\alpha)^2} \right]\prod_{i=1}^{M} \frac{k_2(1-\alpha)^2}{B_i} .
		\end{equation}
		It can be seen that each term in (\ref{eqn:determinant}) is positive, which implies a positive determinant. Then we need to show that  ${\rm adj}({\bf M})$ has positive entries on the diagonal and negative entries elsewhere. On the diagonal, we have  ${\rm adj}({\bf M})_{i,i} = \det({\bf M}_{-i,i})$, where ${\bf M}_{-i,i}$ is the matrix ${\bf M}$ eliminating the $i$th row and column ($i,i$ minor of ${\bf M}$).  It can be seen that ${\bf M}_{-(i,i)}$ is of the similar form of ${\bf M}$ and has a positive determinant. 
		
		For other entries of the adjugate matrix, we have ${\rm adj}({\bf M})_{i,j} = (-1)^{i+j}\det({\bf M}_{-(j,i)}),i\ne j$.  We want to use elementary row operations to transform the matrix ${\bf M}_{-(j,i)}$ into the following form 
		\begin{equation}
		\label{eqn:matrix_determinant}
		{\bf M}'_{-(j,i)}=\left[ {\begin{array}{*{20}{c}}
			G & {G} & {...} &{...}& {G}  \\
			{G} & {H_1} & {...} & & {G}  \\
			{\vdots} & {} & {H_2} & & {\vdots}  \\
			{G} & {G} & {} & {\ddots} & {G}  \\
			{G} & {G} & {...}& {...} & {H_{M-1}}  \\
			\end{array}} \right],
		\end{equation}
		where $G = \frac{k_2 \alpha^2}{W}+k_1$, and $H_i$ is of the form $G+R_i$ with $R_i>0$. It can be verified that the determinant of a matrix of form (\ref{eqn:matrix_determinant}) is strictly positive. Then we need to determine how many elementary row operations we need to transform the matrix ${\bf M}_{-(j,i)}$ into ${\bf M}'_{-(j,i)}$. It can be verified that 
		\begin{enumerate}
			\item  $i+j$ is even and $i\ne 1$,  $j\ne 1$, 3 elementary row operations are needed;
			\item  $i+j$ is even and $i=1$ or  $j=1$, one elementary row operation is needed;
			\item  $i+j$ is odd and $i\ne 1$,  $j\ne 1$, 2 elementary row operations are needed;
			\item  $i+j$ is odd and $i=1$ or  $j=1$, no elementary row operations are needed;
		\end{enumerate}
		In any of the four cases, we end up with ${\rm adj}({\bf M})_{i,j}<0, i\ne j$. 
		As a result we can conclude that the game is supermodular.
		
	\end{proof}

	\section{Proof of Theorem \ref{thm:linear_alpha}}
	\begin{proof}
		First consider an auxiliary case with $M$ incumbent SPs and one entrant SP. In this case, if bundling is not adopted, the price on unlicensed band goes to $0$.  Let $p_i^*$ be the equilibrium price and $x_i ^ *$ be the resulting customer mass on licensed band of incumbent $i$  in the unbundled case.  Let $w_t$ be the total customer mass on the unlicensed band. Let $\alpha_0 = \frac{w_t}{\sum\limits_{i\in\mathbb{I}} x_i +w_t} $, then with similar argument as proving Theorem \ref{thm:incumbent_entrant2}, it can be shown that when  $\alpha < \alpha_0$, the profit of incumbent SPs is higher in the bundling case than the unbundled case. Details are omitted here. 
		
		Then we need to show that when the entrant SP leaves the market, the profit of  incumbent SPs is non decreasing for the bundling case. When the entrant SP leaves the market,  it can be viewed as that the entrant SP is announcing a very high price so that is serving no customers. Given the supermodular property of the game, the best response of any one of the incumbent SPs $i$ is to announce a higher price. Given the supermodular property, other incumbents should also increase their announced prices. Since the increasing price of other incumbents only makes the profit of SP $i$ non-decreasing. As a result each incumbent SP's profit is non-decreasing.
	\end{proof}
	
	\section{Proof of Theorem \ref{thm:expanding band}}
	\begin{proof}
		When $W\to \infty$, we can the congestion of subscribing to SP $i$ for bundling case can be written as $(1-\alpha)g\left(\frac{(1-\alpha)x_i}{B_i}\right)$. When the congestion function is is of the form $g(x) = kx^p$, it becomes $k(1-\alpha)\left(\frac{(1-\alpha)x_i}{B_i}\right)^p$ which is equivalent to expanding the licensed band by a factor $\frac{1}{(1-\alpha)^{p+1}}$ when unlicensed band is not available. 
	\end{proof}
	
	\section{Proof of Theorem \ref{thm:n_incumbent_sw}}
	\begin{proof}
		Let the total licensed bandwidth be $B_t$. In the symmetric case, each SP has as licensed bandwidth $\frac{B_t}{W}$. In the symmetric case with no entrant SP, each incumbent is able to announce a positive price to serve the customers. We can find the equilibrium price for each incumbent SP,
		\begin{equation}
		\label{eqn:equilibrium_price_n_incumbent}
		p = \frac{(1-\alpha)^2W}{2(1-\alpha)^2W+\frac{M-1}{M}(\alpha^2B_t+B_tW)}.
		\end{equation}
		As a result, the total customer mass served by the SPs is
		\begin{small}
			\begin{eqnarray}		
			Q& =& \frac{BW}{\left[2(1-\alpha)^2W+\frac{M-1}{M}(\alpha^2B_t+B_tW)\right]} \nonumber \\
			\label{eqn:equilibrium_customer_n_incumbent2}
			&& \times 
			\left[\frac{M-1}{M}+\frac{(1-\alpha)^2W}{M((1-\alpha)^2W+\alpha^2B_t+B_tW)}\right].
			\end{eqnarray}
		\end{small}
		It is not hard to show that the announced price in (\ref{eqn:equilibrium_price_n_incumbent}) is increasing in $W$ by dividing both the numerator and denominator by $W$. And for the customer mass in (\ref{eqn:equilibrium_customer_n_incumbent2}), note that it is the product of two parts which are both increasing with $W$. That implies that the customer mass served also increases with $W$. Consequently, the profit gained by the SPs should be increasing with $W$. In addition, because the customer surplus in increasing with the total customer mass served, it should also be increasing with $W$. The social welfare is the sum of customer surplus and profit of SPs, as  a result it is also increasing with $W$. 
	\end{proof}
	
	\section{Proof of Theorem \ref{thm:alpha_one_in_one_en}}
	\begin{proof}
		In the linear case with $W\to \infty$, we can find the equilibrium announced price and the resulting customer mass served by the incumbent SP,
		\begin{small}
			\begin{equation}
			\label{eqn:control_alpha_equilibrium_rice}
			\lim\limits_{W\to\infty}p_1 = \frac{(1-\alpha)^2}{4(1-\alpha)^2+3B},\lim\limits_{W\to\infty}x_1 = \frac{B}{4(1-\alpha)^2+3B}.
			\end{equation} 
		\end{small}
		So the profit of incumbent SP in this case is 
		\begin{equation}
		\label{eqn:control_alpha_equilibrium_profit}
		\lim\limits_{W\to\infty}p_1 x_1 =\frac{B(1-\alpha)^2}{\left[4(1-\alpha)^2+3B\right]^2} = \frac{B}{\left[4(1-\alpha)+\frac{3B}{1-\alpha}\right]^2},
		\end{equation}
		which is obviously a concave function in $\alpha$. The optimal solution is then $\alpha^* = 1-\frac{\sqrt{3B}}{2}$ if $B\le \frac{4}{3}$ and $\alpha^* = 0$ if $B > \frac{4}{3}$.
		
	\end{proof}
	
	\section{Proof of Theorem \ref{thm:optimize_profit}}
	\begin{proof}
		To simplify the calculation,let $ \beta = (1-\alpha)^2$. Then we have 
		\begin{equation}
		\label{eqn:profit_beta}
		\lim\limits_{W\to\infty} \frac{d x_ip_i}{d \beta} = \frac{B_t\left[-2\beta^2-3\beta^2\frac{(M-1)B_t}{M}+B_t^3(\frac{M-1}{M})^2\right]}{\left[2\beta+B_t\frac{M-1}{M}\right]^3 (\beta+B_t)}.
		\end{equation}
		
		Just focus on the numerator of (\ref{eqn:profit_beta}), we can find a function of $\frac{\beta}{B_t}$, 
		\begin{equation}
		\label{eqn:f_beta}
		f\left(\frac{\beta}{B_t}\right) = -2\left(\frac{\beta}{B_t}\right)^3-3\frac{M-1}{M}\left(\frac{\beta}{B_t}\right)^2 +\left(\frac{M-1}{M}\right)^2.
		\end{equation}
		Note that $f\left(\frac{M-1}{2M}\right) \ge 0$ and $f\left(\frac{1}{2}\right) \le 0$, since $f\left(\frac{\beta}{B_t}\right) $ is a continuous function, there must exist some $k^* \in \left[\frac{M-1}{2M},\frac{1}{2}\right]$ such that $f(k^*) = 0$. Let $B_{th} = \frac{1}{k^*}$ 
		
		If $B_t\le B_{th}$, we can find that $\beta^* = \frac{B_t}{B_{th}} \in [0,1]$, which implies a feasible solution for $\alpha=1-\sqrt{\frac{B_t}{B_{th}}}$ that maximizes the profit of incumbents.  The resulting announced price and customer mass served are then 
		\begin{equation}
		\label{eqn:small_B_solution}
		p_i = \frac{k^*}{2k^*+\frac{M-1}{M}} ,\; x_i = \frac{1-p_i}{Mk^*+M},\; \forall i.
		\end{equation}
		
		In linear symmetric case, we have $SW = \frac{1}{2}(Mx_i)^2+Mp_ix_i$. Applying the solutions in (\ref{eqn:small_B_solution}), we have
		\begin{equation}
		\label{eqn:sw_solution}
		SW = \frac{(k^*+\frac{M-1}{M})\left[k^*+\frac{M-1}{M}+2k^{*2}+2k^*\right]}{2(2k^*+\frac{M-1}{M})^2(k^*+1)^2}
		\end{equation}
		
		To simplify the notation, let $\frac{M-1}{M} = t_M$ and we have $t_M\in[\frac{1}{2},1)$. Then equation (\ref{eqn:f_beta}) can be rewritten as 
		\begin{equation}
		\label{eqn:equivalent_sol}
		-2k^{*3}-3t_Mk^{*2}+t_M^2 = 0.
		\end{equation}
		By taking total derivative of (\ref{eqn:equivalent_sol}), we have
		\begin{equation}
		\label{eqn:dk_dt}
		\frac{d k^*}{dt_M} = \frac{2t_M-3k^{*2}}{6k^{*2}+6k^*t_M}
		\end{equation}
		Using (\ref{eqn:sw_solution}) and (\ref{eqn:dk_dt}), we can find the derivative of social welfare with respect to $t_M$
		\begin{small}
			\begin{eqnarray}
			\label{eqn:sw_t}
			\frac{d SW}{dt_M}=&&\frac{-(6k^{*4}+18k^{*2}+2k^*)(-2k^{*3}-3t_Mk^{*2}+t_M^2)}{6k^*(1+k^*)^3(k^*+t_M)(2k^*+t_M)^3} \\
			&&-\frac{6k^{*6}(3t_M-1)+k^{*4}(31t_M-1)+5k^{*3}t_M+12k^{*7}}{6k^*(1+k^*)^3(k^*+t_M)(2k^*+t_M)^3}\nonumber\\
			&&- \frac{12k^{*7}+3k^{*5}(6-t_M)+8k^*t_M^3+t_M^3 (2t_M-3k^{*2})}{6k^*(1+k^*)^3(k^*+t_M)(2k^*+t_M)^3}\nonumber
			\end{eqnarray}
		\end{small}
		By equation (\ref{eqn:equivalent_sol}), the first part in (\ref{eqn:sw_t}) is zero and it can be easily verified the rest part of (\ref{eqn:sw_t}) is negative. So the social welfare decreases with the $t_M$. Since $t_M$ is an increasing function in $M$, social welfare decreases with $M$.
		
		On the other hand, if $B_t>B_{th}$, we must have $\beta^* = 1$ or $\alpha^* = 0$, because the derivative is always increasing with $\beta$. The resulting announced price and customer mass are 
		\begin{equation}
		\label{eqn:large_B_solution}
		p_i = \frac{1}{2+B_t\frac{M-1}{M}},\; x_i = \frac{1-p_i}{\frac{M}{B_t}+M},\;\forall i.
		\end{equation}
		Similarly, we can let $t_M= \frac{M-1}{M}$ and find the derivative
		\begin{equation}
		\label{eqn:sw_large_B}
		\frac{d SW}{dt_M}=\frac{B_t^3(1-t_M)}{(1+B_t)^2(2+B_tt_M)^3}>0
		\end{equation}
		As a result social welfare increases with $M$.
	\end{proof}
	
	\section{Proof of Theorem \ref{thm:sw_gap}}
	\begin{proof}
		When $M\to \infty$, we have $\lim\limits_{W\to\infty} \frac{M-1}{M} = 1$ and the solution to (\ref{eqn:f_beta}) becomes $2$.
		
		When $B_t\le2$, the optimal $\alpha$ that maximize the profit is $\alpha = 1-\sqrt{B_t/2}$, the resulting announced price by each incumbent goes to $\frac{1}{4}$ and consequently the social welfare is $SW = \frac{1}{4}$. Since it is easy to find that the social welfare in the social optimal case is $\frac{1}{2}$, the social welfare gap is $Gap = \frac{1}{4}$.
		
		When $B_t>2$, the optimal $\alpha$ becomes $0$. In this case, the resulting social welfare is $SW = \frac{B_t^2+2B_t}{2(2+B_t)^2}$. The social welfare gap is then $Gap = \frac{1}{2+B_t}$. These two case give us a simplified form in equation (\ref{eqn:sw_gap}).
	\end{proof}
	
\end{document}